\documentclass{lmcs}
\pdfoutput=1
\usepackage[utf8]{inputenc}

\usepackage{lastpage}
\lmcsdoi{21}{1}{20}
\lmcsheading{}{\pageref{LastPage}}{}{}%
{Oct.~10,~2023}{Feb.~28,~2025}{}


\usepackage{afterpage}
\usepackage{cite}
\usepackage{amssymb}
\usepackage{amsmath}
\usepackage{bussproofs}
\usepackage{multicol}
\usepackage{comment}
\usepackage{enumitem}
\setlist[enumerate]{itemsep=0mm}
\usepackage{url}
\usepackage{xcolor}
\usepackage{extarrows} 
\usepackage{longtable}
\usepackage{multirow}
\usepackage{tabularx}
\usepackage{framed}
\usepackage{float}
\usepackage[figurename=Figure]{caption}

\usepackage{orcidlink}
\usepackage{hyperref}
 \hypersetup{
     colorlinks=false,
     linkcolor=black,
     filecolor=black,
     citecolor = black,      
     urlcolor=black,
     }
     
\begin{document}
\title[Congruence Closure Modulo Groups]{Congruence Closure Modulo Groups}\thanks{This
research is supported by the Austrian Science Fund (FWF) project I 5943.}

\author[D.~Kim]{Dohan Kim\lmcsorcid{0000-0003-1973-615X}}
\address{Department of Computer Science, University of Innsbruck, Innsbruck, Austria}
\email{dohan.kim@uibk.ac.at}

\begin{abstract}
\noindent This paper presents a new framework for constructing congruence closure of a finite set of ground equations over uninterpreted symbols and interpreted symbols for the group axioms. In this framework, ground equations are flattened into certain forms by introducing new constants, and a completion procedure is performed on ground flat equations. The proposed completion procedure uses equational inference rules and constructs a ground convergent rewrite system for congruence closure with such interpreted symbols. If the completion procedure terminates, then it yields a decision procedure for the word problem for a finite set of ground equations with respect to the group axioms. This paper also provides a sufficient terminating condition of the completion procedure for constructing a ground convergent rewrite system from ground flat equations containing interpreted symbols for the group axioms. In addition, this paper presents a new method for constructing congruence closure of a finite set of ground equations containing interpreted symbols for the semigroup, monoid, and the multiple disjoint sets of group axioms, respectively, using the proposed framework.
\end{abstract}

\maketitle

\section{Introduction}
Congruence closure procedures~\cite{Downey1980,Nelson1980,Kozen1977} have been widely studied for several decades, and play important roles in software and hardware verification systems~\cite{Nelson1980,Cyrluk1995,Sjoberg2015}, satisfiability modulo theories (SMT) solvers~\cite{Barrett2018, Moura2011}, etc. They can be used to determine whether another ground equation is a consequence of a given finite set of ground equations.\\
\indent A rewrite-based congruence closure procedure in the framework of ground completion was proposed by Kapur~\cite{Kapur1997}. It is based on flattening nonflat terms appearing in the input set of ground equations by introducing new constants. In~\cite{Kapur2021,Kapur2023, Tiwari2003}, the associative and commutative ($AC$) congruence algorithms were presented for congruence closure with interpreted symbols satisfying the $AC$ properties.

In~\cite{Kim2021}, Kim and Lynch presented a framework for congruence closure modulo permutation equations (i.e., flat permutative equations). In~\cite{Baader2022}, Baader and Kapur presented \emph{semantic congruence closure} for constructing congruence closure with interpreted symbols axiomatized by \emph{strongly shallow equations}. They also provided a congruence closure algorithm with \emph{extensional} symbols.
\\
\indent Meanwhile, \emph{group theory}~\cite{Hungerford1980} is fundamental in mathematics, and has a wide variety of applications in physics, biology, and computer science. The \emph{group axioms} of a binary operator $f$, a unary operator $i$, and the unit element $1$ are

\begin{longtable}{llll}
$f(f(x,y),z)\approx f(x, f(y,z))$\quad\quad\quad\quad & $f(x, 1)\approx x$ \quad\quad\quad\quad & $f(1, x)\approx x$ &\quad\quad\quad\quad\quad\\
$f(x, i(x))\approx 1$ &$f(i(x), x)\approx 1$\quad\quad\quad\quad & &
\end{longtable}

The well-known convergent rewrite system for the group axioms is as follows~\cite{Chenadec1984}:

\begin{longtable}{lll}
$i(1)\rightarrow 1$  & $f(x,1) \rightarrow x$ &$f(1,x)\rightarrow x$\\ 
$i(i(x))\rightarrow x$ & $f(x, i(x))\rightarrow 1$ & $f(i(x),x)\rightarrow 1$ \\
$f(i(x), f(x,y))\rightarrow y$ & $f(x, f(i(x),y))\rightarrow y$\quad\quad\quad & $i(f(x,y))\rightarrow f(i(y), i(x))$\\
$f(f(x,y),z)\rightarrow f(x, f(y,z))$ & &
\end{longtable}

Above, the associative symbol $f$ has a fixed arity. In this paper, \emph{associative flattening}~\cite{Dershowitz2001, Tiwari2003} w.r.t.\;each associative symbol $f$ is applied, so the arity of $f$ becomes variadic. For example, $f(f(u,v),f(w,z))$ is associatively flattened to $f(u,v,w,z)$ for an associative symbol $f$. Rewriting on associatively flat terms (cf.~\cite{Marche1996}) is also used, which does not need to introduce \emph{extension rules}~\cite{Dershowitz2001} during a completion procedure.\\
\indent This paper presents a new framework for computing congruence closure of a finite set of ground equations over uninterpreted and interpreted symbols for the group axioms $G$. If the proposed completion procedure terminates, then one can decide whether a ground equation follows from a given finite set of ground equations containing the interpreted symbols for $G$.\\
\indent The approach used in this paper is roughly illustrated as follows using a simple example. Let $E=\{f(h(a), h(a))\approx 1, i(h(a))\approx b\}$, where an associative symbol $f$, the inverse symbol $i$, and the unit $1$ are the interpreted symbols for $G$, and $h$ is an uninterpreted symbol. First, $f(h(a), h(a))\approx 1$ and $i(h(a))\approx b$ are flattened into $h(a)\approx c_1$, $f(c_1,c_1)\approx 1$, and $i(c_1)\approx b$ using a new constant $c_1$. Next, for each constant $1$, $a$, $b$, and $c_1$, terms with the inverse symbol $i$ are considered by adding $i(1)\approx 1$, $i(a)\approx c_2$, $i(c_2)\approx a$ (because $i(i(a))\approx a$), and $i(b)\approx c_1$, where $c_2$ is a new constant. The equations $f(a,c_2)\approx 1, f(c_2,a)\approx 1, f(c_1,b)\approx 1, f(b, c_1)\approx 1$ are also added by taking the equations $f(x,i(x))\approx 1$ and $f(i(x),x)\approx 1$ in the group axioms into account. The resulting set $S(E)$ of this procedure from $E=\{f(h(a), h(a))\approx 1, i(h(a))\approx b\}$ is $S(E)=\{h(a) \approx c_1, f(c_1, c_1)\approx 1, i(c_1)\approx b, i(1)\approx 1, i(a)\approx c_2, i(c_2)\approx a, i(b)\approx c_1, f(a,c_2)\approx 1, f(c_2,a)\approx 1, f(c_1,b)\approx 1, f(b,c_1)\approx 1\} \cup U(C)$, where $U(C)$ is the set of ground instantiations of $f(x, 1)\approx x$ and $f(1,x)\approx x$ in the group axioms using the set of constant symbols $C$ only. By adding certain ground flat equations entailed by the group axioms, the proposed completion procedure for constructing congruence closure modulo $G$ w.r.t.\;$E$ is only concerned with ground (flat) equations instead of taking the group axioms (containing variables) into account.\\
\indent Now, the arguments of a term headed by an associative symbol $f$ are represented by the corresponding string for ground flat equations. For example, $f(a,b,c,d)$ is represented by $f(abcd)$, where $f$ is an associative symbol. Then, the proposed completion procedure using string-based equational inference rules is proceeded. Equation $c_1 \approx b$ can be inferred from $f(c_1c_1)\approx 1$ and $f(c_1b)\approx 1$ using an equational inference rule discussed later in this paper. Roughly speaking, $f(c_1c_1b)\approx b$ follows from $f(c_1c_1)\approx 1$, and $f(c_1c_1b)\approx c_1$ follows from $f(c_1b)\approx 1$, so $c_1 \approx b$ can be inferred from $f(c_1c_1)\approx 1$ and $f(c_1b)\approx 1$. The completion terminates with $S_\infty(E)=\{h(a) \approx b, f(bb)\approx 1,  i(b)\approx b, i(1)\approx 1, i(a)\approx c_2, i(c_2)\approx a, f(ac_2)\approx 1, f(c_2a)\approx 1, c_1 \approx b\} \cup \bar{U}(C)$, where $c_1 \succ b$ and $U(C)$ is updated to $\bar{U}(C)$ by rewriting each occurrence of $c_1$ in $U(C)$ to $b$. Now, the rewrite system $S_\infty^\succ(E)$ can be obtained by orienting each equation in $S_\infty(E)$ into the rewrite rule using $\succ$ (defined in Section~\ref{sec:cc}). Then each $f(u)$ for string $u=c_ic_{i+1}\cdots c_{i+j}$ in the rewrite system is restored into $f(c_i, c_{i+1}, \ldots, c_{i+j})$. The ground rewrite system for groups $R(G)$ (defined in Section~\ref{sec:preliminaries}) combined with $S_\infty^\succ(E)$ modulo associativity can decide whether a ground equation follows from $E$ w.r.t.\;$G$. For example, $h(a)\approx i(i(b))$ follows from $E$ w.r.t.\;$G$, where $h(a)$ can be rewritten to $b$ and $i(i(b))$ can be rewritten to $b$ using the rewrite system $S_\infty^\succ(E)$.\\
\indent The key insight of the proposed approach is that only certain ground equations entailed by the group axioms are added for the completion procedure, while taking only ground flat equations into account during its entire completion procedure. This keeps the completion procedure from interacting with the (nonground) convergent rewrite system for groups directly. Now, a completion procedure for groups~\cite{Holt2005,Sims1994} using their monoid presentations is extended to the proposed approach in which the arguments of each term headed by an associativity symbol are represented by a string. (Roughly speaking, a \emph{monoid presentation} of a group is represented by string relations and generators, adding the relations of the form $xx^{-1} \approx 1$ and $x^{-1} x \approx 1$ for each $x \in X$ to the relations of a presentation of the group, where $X$ is the set of generators of the presentation of the group. See~\cite{Holt2005} for details.)\\
\indent Since the word problem for finitely presented groups is undecidable in general~\cite{Holt2005}, the word problem for a finite set of ground equations $E$ w.r.t.\;$G$ is undecidable in general too. The proposed completion procedure may not terminate and yield an infinite convergent rewrite system for congruence closure of $E$ w.r.t.\;$G$. If it terminates, then it provides a decision procedure for the (ground) word problem for $E$ w.r.t.\;$G$.\\
\indent In addition, this paper discusses a sufficient terminating condition of the proposed completion procedure by attempting to associate a finite set of ground flat equations derived from $E$ with a monoid presentation of a group. If it is a monoid presentation of a finite group, then it necessarily terminates and yields a finite ground convergent rewrite system for congruence closure of $E$ w.r.t.\;$G$, providing a decision procedure for the word problem for $E$ w.r.t.\;$G$.

Based on the proposed framework, this paper also presents a new approach to constructing a rewriting-based congruence closure of a finite set of ground equations w.r.t.\;the semigroup, monoid, and the multiple disjoint sets of group axioms, respectively. Interestingly, it shares the same proposed ground completion procedure and yields a (possibly infinite) ground convergent rewrite system for congruence closure of a finite set of ground equations w.r.t.\;the semigroup, monoid, and the multiple disjoint sets of group axioms, respectively. 

\section{Preliminaries}\label{sec:preliminaries}

The reader is assumed to have some familiarity with rewrite systems~\cite{Baader1998,Dershowitz2001}. This paper refers to~\cite{Kapur1997,Kapur2021,Baader2020,Baader2022,Tiwari2003} for the definitions and notations of congruence closure. 

Let $\Sigma$ be a finite set of function symbols of arity $\geq 1$ and $C_0$ be a finite set constant symbols (i.e.,\;function symbols of arity 0). We denote by $\Sigma_A$ the subset of $\Sigma$ containing all the associative function symbols. We also denote by $C\supseteq C_0$ a finite set of constant symbols that may include (finitely many) new constant symbols other than the (original) constant symbols in $C_0$.

\emph{Terms} are built from variables, constants (in $C$), and function symbols (in $\Sigma$).  A \emph{ground term} is a term not containing variables. We denote by $T(\Sigma, C_0)$ (resp.\;$T(\Sigma, C)$) the set of ground terms built from $\Sigma$ and $C_0$ (resp.\;$\Sigma$ and $C$). The symbols $x,y,z,\ldots$ along with their subscripts are used to  denote variables.

An \emph{equation} is an unordered pair of terms, written $s\approx t$.

Function symbols that require to satisfy additional \emph{semantic properties} (often expressed by a set of nonground axioms) are called \emph{interpreted symbols}, while function symbols that do not require such properties are called \emph{uninterpreted symbols}.

For a fixed arity binary function symbol $f\in \Sigma_A$, $f(f(x,y),z)\approx f(x, f(y,z))$ defines the theory of associativity. Alternatively, the set of all equations $A$ of the following form defines the theory of associativity for variadic terms (see~\cite{Tiwari2003}).\\

$f(x_1,\ldots,x_m, f(y_1,\ldots, y_i), z_1,\ldots, z_n) \approx f(x_1,\ldots, x_m, y_1,\ldots, y_i, z_1,\ldots, z_n)$,\\

\noindent where $f\in \Sigma_A$ and $\{m+n+1, m+n+i, i\}\subset \alpha(f)$ and $\alpha(f)=\{2,3,4,\ldots\}$. (Here, $\alpha$ denotes the arity function for $\Sigma$.) By \emph{associatively flattening} of a term, we mean the normalization of a term w.r.t.\;$A$ considered as a rewrite system in such a way that each equation $f(x_1,\ldots,x_m, f(y_1,\ldots, y_i), z_1,\ldots, z_n) \approx f(x_1,\ldots, x_m, y_1,\ldots, y_i, z_1,\ldots, z_n)$ in $A$ is oriented as a rewrite rule $f(x_1,\ldots,x_m, f(y_1,\ldots, y_i), z_1,\ldots, z_n) \rightarrow f(x_1,\ldots, x_m, y_1,\ldots, y_i, z_1,\ldots, z_n)$.  For example, $f(f(u,v),f(w,z))$ is associatively flattened to $f(u,v,w,z)$ for $f\in \Sigma_A$.

In the remainder of this paper, the \emph{associativity axioms} (\emph{semigroup axioms}) are the equations in $A$, the \emph{unit axioms} are the equations $f(x,1) \approx x$ and $f(1,x)\approx x$, the \emph{inverse axioms} are the equations $f(x,i(x))\approx 1$ and $f(i(x),x)\approx 1$, the \emph{monoid axioms} are the equations in $A\cup \{f(x,1)\approx x, f(1,x)\approx x\}$, and the \emph{group axioms} are the equations in $A\cup \{f(x,1)\approx x, f(1,x)\approx x, f(x,i(x))\approx 1, f(i(x),x)\approx 1\}$ for some $f\in \Sigma_A$. Unless otherwise stated, associative flattening w.r.t.\;each associative symbol $f\in \Sigma_A$ is always applied, and by $G$, we mean the set of group axioms.

Given a (finite) set of constant symbols $C$, we denote by $U(C):=\{f(c, 1) \approx c\,|\, c \in C\} \cup \{f(1, c) \approx c \,|\, c \in C\}$ the set of ground instantiations of $f (x, 1) \approx x$ and $f (1, x) \approx x$ in $G$ using only the elements in $C$.

The \emph{equational theory} induced by a set of equations $E$, denoted by $\approx_E$, is the smallest congruence that is closed under substitutions and contains $E$ (i.e., the smallest congruence containing all instances of the equations of $E$). Two terms $s$ and $t$ are \emph{equal modulo} $A$ (associativity), denoted by $s\xleftrightarrow{*}_{A}t$, iff their associatively flat forms are equal. Each associatively flat term represents its $A$-equivalence class.

A term $t$ is said to $A$-\emph{match} another term $s$ (and $s$ is called an $A$-\emph{instance} of $t$) iff there is a substitution $\sigma$ such that $s \xleftrightarrow{*}_A t\sigma$. Note that if both $s$ and $t$ are ground terms in $T(\Sigma, C)$, then $s$ $A$-matches $t$ iff $s \xleftrightarrow{*}_A t$.

The \emph{depth} of a term $t$ is defined inductively as $depth(t) = 0$ if $t$ is a variable or a constant and $depth(f(s_1,\ldots,s_n)) = 1 + \text{max}\{depth(s_i)\,|\,1\leq i\leq n\}$. 

The \emph{size} of a term $t$, denoted by $|t|$, is the total number of occurrences of function and variable symbols in it. 

A \emph{rewrite system} is a set of rewrite rules. A rewrite relation $\rightarrow_{R/A}$  is defined as follows: $u\rightarrow_{R/A} v$ if there are terms $s$ and $t$ such that $u\approx_A s$, $s\rightarrow_R t$, and $t\approx_A v$. For example, if $f\in \Sigma_A$ and $f(b,c) \rightarrow e \in R$, then we have $f(a,b,c,d)\rightarrow_{R/A} f(a,e,d)$ because $f(a,b,c,d)\approx_A f(a,f(f(b,c),d)) \rightarrow_R f(a, f(e,d))\approx_A f(a,e,d)$.

The transitive and reflexive closure of  $\rightarrow_{R/A}$ is denoted by $\xrightarrow{*}_{R/A}$.

We simply denote by $R/A$ the rewrite relation $\rightarrow_{R/A}$. 

We say that a term $t$ is in $R/A$-\emph{normal form} if there is no term $t^\prime$ such that $t \rightarrow_{R/A} t^\prime$. Otherwise, we say that $t$ is $R/A$-\emph{reducible}.

A rewrite relation $R/A$ is \emph{confluent} if for all terms $r$,$s$ and $t$ with $r \xrightarrow{*}_{R/A}s$ and $r \xrightarrow{*}_{R/A}t$, there are terms $u$ an $v$ such that $s \xrightarrow{*}_{R/A} u \xleftrightarrow{*}_{A} v   \xleftarrow{*}_{R/A} t$. (Recall that each $A$-equivalence class is represented by a single associatively flat term.)

A rewrite relation $R/A$ is \emph{locally confluent} if for all terms $r$, $s$ and $t$ with $r \rightarrow_{R/A}s$ and $r \rightarrow_{R/A}t$, there are terms $u$ and $v$ such that $s \xrightarrow{*}_{R/A} u \xleftrightarrow{*}_A v \xleftarrow{*}_{R/A} t$. (In this paper, $\rightarrow_{R/A}$ is used to rewrite associatively flat ground terms and $A$-matching for $\rightarrow_{R/A}$ can be done using associatively flat ground terms. Also, the check for local confluence is simpler, i.e., $R/A$ is locally confluent on associatively flat ground terms if for all associatively flat ground terms $r$, $s$ and $t$ with $r \rightarrow_{R/A}s$ and $r \rightarrow_{R/A}t$, there is an associatively flat ground term $u$ such that $s \xrightarrow{*}_{R/A} u    \xleftarrow{*}_{R/A} t$.)

A rewrite relation $R/A$ is \emph{terminating} if there is no infinite sequence $t_0 \rightarrow_{R/A}t_1 \rightarrow_{R/A} t_2 \rightarrow_{R/A}\cdots.$

A terminating rewrite relation $R/A$ is confluent iff it is locally confluent.

A rewrite system $R$ is \emph{convergent modulo A} if the rewrite relation $\rightarrow_{R/A}$ is confluent and terminating.

We denote by $\Sigma_G$ the interpreted symbols for the set of group axioms $G:=A\cup \{f(x,1)\approx x, f(1,x)\approx x, f(x, i(x))\approx 1, f(i(x),x)\approx 1\}$. Here, we have $f, i, 1\in \Sigma_G$, where $f\in \Sigma_A$, $i \in \Sigma$, and $1\in C_0$.\\
\indent The convergent rewrite system for the set of group axioms $G$, denoted $R(G)$, on associatively flat terms is given as follows:\\

\begin{tabular}{llll}
$i(1)\rightarrow 1$  & $f(x,1) \rightarrow x$ & $f(1,x)\rightarrow x$ & $i(i(x))\rightarrow x$\\ 
$f(x, i(x))\rightarrow 1$ & $f(i(x),x)\rightarrow 1$ & $i(f(x,y))\rightarrow f(i(y), i(x))$\\\\
\end{tabular}

\noindent where $f\in \Sigma_G$.
Note that the extension rule~\cite{Dershowitz2001} $f(x, i(x), y)\rightarrow y$ of $f(x, i(x))\rightarrow 1$ and the extension rule $f(i(x),x,y)\rightarrow y$ of $f(i(x),x)\rightarrow 1$ can be simplified and are not needed by $R(G)/A$. In fact, extension rules can be dealt implicitly without explicitly generating them for rewriting on associatively flat terms. This is further discussed in the next section.

\begin{lem}\label{lem:canonical1} The rewrite system $R(G)$ is convergent modulo $A$.
\end{lem}

\indent Given a finite set $E=\{s_i \approx t_i\,|\,1\leq i \leq m\}$ of ground equations where $s_i, t_i \in T(\Sigma, C_0)$, the \emph{congruence closure} $CC(E)$~\cite{Kapur2021,Baader2020} is the smallest subset of $T(\Sigma, C_0)\times T(\Sigma, C_0)$ containing $E$ and is closed under the following rules: (i) for every $s\in  T(\Sigma, C_0)$, $s\approx s \in CC(E)$ (\emph{reflexivity}), (ii) if $s\approx t \in CC(E)$, then $t\approx s \in CC(E)$ (\emph{symmetry}), (iii) if $s\approx t\in CC(E)$ and $t\approx u \in CC(E)$, then $s\approx u \in CC(E)$ (\emph{transitivity}), and (iv) if $f\in \Sigma$ is an $n$-ary function symbol ($n>0$) and $s_1 \approx t_1,\ldots, s_n \approx t_n \in CC(E)$, then $f(s_1, \ldots, s_n) \approx f(t_1, \ldots, t_n)\in CC(E)$ (\emph{congruence}). Note that $CC(E)$ is also the equational theory induced by $E$.

\indent Given a finite set $E$ of ground equations between terms in $T(\Sigma, C_0)$, we also consider a set of nonground equations $F$ such that the equations in $F$ are also satisfied as well as the equations in $E$. Given $s,t \in T(\Sigma, C_0)$, $s\approx t$ \emph{follows from} $E$ w.r.t.\;$F$ (written $s\approx_E^F t$) if $s^\mathcal{A} = t^\mathcal{A}$ holds in all models $\mathcal{A}$ of $E\cup F$. Therefore, $\approx_E^F$ is the restriction of $\approx_{E\cup F}$ (i.e.,\;the equational theory induced by $E\cup F$) to the ground terms in $T(\Sigma, C_0)$ (or $T(\Sigma, C)$ by a slight abuse of notation). The \emph{word problem} for $E$ w.r.t.\;$F$ is the relation $\approx_E^F$.

The \emph{congruence closure of} $E$ \emph{w.r.t.}\;$F$, denoted by $CC^F(E)$, is the smallest subset of $T(\Sigma, C_0)\times T(\Sigma, C_0)$ that contains $E$ and is closed under reflexivity, symmetry, transitivity, and congruence, and the following rule:\\

\noindent $-$ if $l\approx r \in F$ and $\sigma$ is a substitution that maps the variables in $l$, $r$ to elements of $T(\Sigma, C_0)$, then $l\sigma \approx r\sigma \in CC^F(E)$.\\

\indent Birkhoff's theorem says that $\approx_F$ coincides with $\stackrel{*}{\leftrightarrow}_F$, where $\stackrel{*}{\leftrightarrow}_F$ denotes the reflexive, transitive, and symmetric closure of $\rightarrow_F$. Observe that $CC^G(E)$ is the equational theory induced by $E\cup G$, which coincides with $\approx_E^G$ by Birkhoff's theorem. We also say $CC^G(E)$ as \emph{congruence closure modulo} $G$ for $E$.

\section{Congruence Closure Modulo Groups}\label{sec:cc}
In this section, we denote by $E$ a finite set of ground equations $s\approx t$ between (ground) terms $s,t\in T(\Sigma, C_0)$, and by $G$ the set of group axioms $A\cup \{f(x,1)\approx x, f(1,x)\approx x, f(x, i(x))\approx 1, f(i(x),x)\approx 1\}$. We also denote by $W$ an infinite set of constants taken from $\{c_1,c_2, \ldots\}$, where $W$ is disjoint from $C_0$, $C_1$ a finite subset chosen from $W$, and $C:=C_0\cup C_1$.

\subsection{Ground Completion using Ground Flat Equations}
\begin{defi}\label{defn:gequations}\normalfont
By a \emph{ground (fully) flat equation}, we mean an (associatively flattened) ground equation over $T(\Sigma, C)$ in one of the following three forms:
\begin{enumerate}[(i)]
\item A \emph{C}-\emph{constant equation} is an equation of the form $c_i\approx c_j$, where $c_i$ and $c_j$ are constants in $C$.
\item A \emph{D}-\emph{flat equation} is an equation of the form $g(c_1,\ldots,c_n) \approx c$, where $c_1,\ldots,c_n, c$ are constants in $C$, $g$ is an $n$-ary function symbol with $n\geq 1$.\footnote{If $g\in \Sigma_A$, then $g(c_1,\ldots,c_n)$ is an associatively flat ground term of depth 1.}
\item An \emph{A}-\emph{flat equation} is an equation of the form $f(c_1,\ldots,c_m) \approx f(d_1,\ldots, d_n)$, where $f \in \Sigma_A$ and $c_1,\ldots, c_m, d_1,\ldots, d_n$ are constants in $C$.
\end{enumerate}
By \emph{ground (fully) flat terms}, we mean terms occurring in ground (fully) flat equations over $T(\Sigma, C)$.
\end{defi}

The proposed approach to constructing congruence closure modulo $G$ for $E$ using a ground completion procedure consists of three phases. Roughly speaking, the first phase transforms a set of ground equations $E$ into the set of ground flat equations $E'$ (cf.~\cite{Kapur1997,Tiwari2003}). The second phase adds certain ground flat equations entailed by $G$ using $E'$. Call $S(E)$ the resulting set. Finally, the last phase applies a ground completion procedure on $S(E)$.

The key advantages of using ground flat equations are as follows: (i) the ground completion procedure is simple by taking only ground flat equations into account during its entire procedure, (ii) no complex ordering is needed, and (iii) the overlaps between ground flat equations and $R(G)$ are in restricted form, which simplifies the construction of a ground convergent rewrite system for congruence closure modulo $G$ for $E$. In fact, by adding certain ground flat equations entailed by $G$, the proposed completion procedure does not interact with $R(G)$ at all and performs only on $S(E)$. In what follows, each phase is discussed in detail. We assume that multiple associative symbols are allowed (i.e.,\;$|\Sigma_A| \geq 1$), but only one of them is used for $G$.\\
 
\noindent {\bf Phase I}: Given a finite set of ground equations $E$ between terms $s,t \in T(\Sigma, C_0)$,
\begin{enumerate}
\item Perform associative flattening to each term in $E$. Call $E_1$ the output of this step, where all terms in $E_1$ are fully associatively flattened.
\item Normalize each term in $E_1$ using the rewriting by $R(G)/A$. Call $E_2$ the output of this step, where all terms in $E_2$ are in $R(G)/A$-normal form and the trivial equations (i.e.,\;the equations of the form $s\approx s$) are all removed.
\item Flatten nonflat terms in $E_2$ using new constants from $W:=\{c_1,c_2, \ldots\}$ in an iterative way so that the output of this step consists of only constant, $D$-flat, and $A$-flat equations. Here, each $A$-flat equation is not allowed to be further transformed, for example, into $D$-flat equations. In this step, if an equation with nonconstant flat terms is neither an $A$-flat equation nor a $D$-flat equation (e.g.\;$f(c_1) \approx g(c_2)$ with $f \in \Sigma_A $ and $g\notin  \Sigma_A$), then it is further transformed into $D$-flat (or $D$-flat and constant) equations using new constants from $W$. 
\end{enumerate}

The output of Phase I is denoted by $E'$, where all equations in $E'$ are constant, $D$-flat, or $A$-flat equations. Also, $C:=C_0 \cup C_1$ is a finite set of constants obtained after Phase I, where $C_1$ is a set of new constants introduced by flattening in step 3 of Phase I.

\begin{exa}\label{ex:ex1}\normalfont Let $E$ =$\{f(a,a)\approx f(h(a), f(i(h(a)), 1)), f(a, h(a))\approx b, f(i(a), b)\approx b\}$, where $f\in \Sigma_G$ is an associative symbol, $i\in\Sigma_G$ is the inverse symbol, and $1\in \Sigma_G$ is the unit for $G$.

For step 1 of Phase I, perform associative flattening to $f(h(a), f(i(h(a)), 1))$, which yields $f(h(a), i(h(a)), 1)$. Then $E_1$ is $E$ with $f(a,a)\approx f(h(a), f(i(h(a)), 1))$ being replaced by  $f(a,a)\approx f(h(a), i(h(a)), 1)$.

For step 2, term $f(h(a), i(h(a)), 1)$ is normalized to 1 by $R(G)/A$. Therefore, $E_2=\{f(a,a)\approx 1, f(a, h(a))\approx b, f(i(a), b)\approx b\}$.

Finally, for step 3,  new constants $c_1$ and $c_2$ for $h(a)$ and $i(a)$ are introduced, respectively. Now, $E'=\{h(a)\approx c_1, i(a)\approx c_2, f(a,a)\approx 1, f(a, c_1)\approx b, f(c_2, b)\approx b\}$ and $C=\{1, a, b, c_1, c_2\}$.
\end{exa}

Next, the purpose of Phase II is to add certain ground flat equations entailed by $G$ using $E'$. Since $R(G)$ contains variables, this attempts to keep the ground flat equations in $S(E)$ (output of Phase II) from interacting with $R(G)$ during a completion procedure. In the following, $f\in \Sigma_G$ is an associative symbol, $i\in \Sigma_G$ is the inverse symbol, and $1\in \Sigma_G$ is the unit in $G$.\\

\noindent {\bf Phase II}: Given $C$ and $E'$ obtained from $E$ by Phase I, copy $C$ to $C'$ and $E'$ to $E^{''}$:
\begin{enumerate}
\item For each constant $c_k\in C'$ and $c_k \neq 1$, repeat the following step:

\noindent If neither $i(c_k)\approx c_i$ nor $i(c_j)\approx c_k$ appears in $E^{''}$ for some $c_i,c_j\in C'$, then $ E':=E'\cup \{i(c_k)\approx c_m\}$ and $C:=C\cup \{c_m\}$ for a new constant $c_m$ taken from $W$.
\item Set $I(E'):=\{i(1)\approx 1\} \cup \{i(c_n)\approx c_m\,|\,i(c_m)\approx c_n\in E'\} \cup \{f(c_m, c_n)\approx 1\,|\,i(c_m)\approx c_n\in E'\}\cup \{f(c_n, c_m)\approx 1\,|\,i(c_m)\approx c_n\in E'\}$.
\item Set $S(E):= E' \cup I(E') \cup U(C)$ and return $S(E)$, where $U(C):=\{f(c,1)\approx c\,|\,c\in C\} \cup \{f(1,c)\approx c\,|\,c\in C\}$.
\end{enumerate}

Now, $S(E)$ is the output of Phase II. (Also, $C$ can be updated during Phase II because new constants can be added in step 1 of Phase II.) No new constant is added to $C$ after Phase II.

\begin{exa}\label{ex:ex2}\normalfont (Continued from Example~\ref{ex:ex1}) After Phase I, we have $C=\{1, a, b, c_1, c_2\}$ and $E'=\{h(a)\approx c_1, i(a)\approx c_2, f(a,a)\approx 1, f(a,c_1)\approx b, f(c_2,b)\approx b\}$, where $f\in \Sigma_G$.\\
\indent For step 1 of Phase II, add $i(b)\approx c_3$ and $i(c_1)\approx c_4$ to $E'$, where $c_3$ and $c_4$ are the new constants taken from $W$. We now have $C=\{1,a,b,c_1,c_2,c_3,c_4\}$ and $E'=\{h(a)\approx c_1, i(a)\approx c_2, f(a,a)\approx 1, f(a,c_1)\approx b, f(c_2,b)\approx b, i(b)\approx c_3, i(c_1)\approx c_4\}$. For step 2, we have $I(E')=\{i(1)\approx 1, i(c_2)\approx a,i(c_3)\approx b,i(c_4)\approx c_1,f(a,c_2)\approx 1, f(c_2,a)\approx 1,f(b,c_3)\approx 1, f(c_3, b)\approx 1, f(c_1, c_4)\approx 1, f(c_4, c_1)\approx 1\}$. For step 3, we have $U(C)=\{f(1,1)\approx 1, f(a,1)\approx a, f(1,a)\approx a, f(b,1)\approx b, f(1,b)\approx b, f(c_1, 1)\approx c_1, f(1,c_1)\approx c_1, f(c_2, 1)\approx c_2, f(1,c_2)\approx c_2, f(c_3, 1)\approx c_3, f(1,c_3)\approx c_3,f(c_4, 1)\approx c_4, f(1,c_4)\approx c_4\}$. Finally, we have $S(E)=E'\cup I(E') \cup U(C)$ and return $S(E)$.
\end{exa}

The proof of the following lemma is adapted from the proof of Lemma 3 in~\cite{Baader2022}.

\begin{lem}\label{lem:conservative} $S(E)$ w.r.t.\;$G$ is a conservative extension of $E$ w.r.t.\;$G$, i.e., $s_0\approx_E^G t_0$ iff $s_0\approx_{S(E)}^G t_0$ for all terms $s_0, t_0 \in T(\Sigma, C_0)$.
\end{lem}
\begin{proof}
To show the \emph{if}-\emph{direction},  assume that $s_0\not\approx_E^G t_0$. Then, there is some algebra $\mathcal{A}$ over the signature $\Sigma \cup C_0$ satisfying all equations in $E\cup G$ such that $s_0^{\mathcal{A}}\neq t_0^{\mathcal{A}}$. Let $Sub(E)$ be the set of subterms of the terms occurring in $E$.
First, observe that for each new constant $c_k \in C_1$, there is some term $s_k \in Sub(E) \setminus C_0$ such that $c_k \approx_{S(E)} s_k$. 
Now, we may expand $\mathcal{A}$ to the new constants in $C_1$ by interpreting each $c_k \in C_1$ as $s_k^\mathcal{A}$ for some term $s_k \in Sub(E) \setminus C_0$. We call this expanded algebra $\mathcal{B}$. Since $s_0^{\mathcal{B}} = s_0^{\mathcal{A}} \neq t_0^{\mathcal{A}} = t_0^{\mathcal{B}}$, it is sufficient to show that $\mathcal{B}$ satisfies each equation in $S(E) \cup G$. We know that $\mathcal{A}$ satisfies $G$, so its expansion $\mathcal{B}$ obviously satisfies $G$ too. It remains to show that $\mathcal{B}$ satisfies $S(E)$. First, consider a $C$-constant equation $c_i \approx c_j \in S(E)$. Then, we have $s_i \approx s_j \in E$. Here, $s_i$ (resp.\;$s_j$) is simply $c_i$ (resp.\;$c_j$) if $c_i \in C_0$ (resp.\;$c_j \in C_0)$. Now, we have $c_i^{\mathcal{B}} = s_i^{\mathcal{A}} = s_j^{\mathcal{A}} = c_j^{\mathcal{B}}$, and thus $\mathcal{B}$ satisfies $c_i \approx c_j$.

Next, consider an $A$-flat equation $f(c_1,\ldots,c_m) \approx f(c_{m+1},\ldots,c_n) \in S(E)$, where $f\in \Sigma_A$ and $c_1,\ldots,c_n \in C$. Then we have $f(s_1,\ldots,s_m) \approx f(s_{m+1},\ldots,s_n) \in E$ for some $s_1,\ldots s_n\in Sub(E)$. Again, $s_k$ is simply $c_k$ if $s_k \in C_0$. Now, we have $f^{\mathcal{B}}(c_1^{\mathcal{B}},\ldots,c_m^{\mathcal{B}})=f^{\mathcal{A}}(s_1^{\mathcal{A}},\ldots,s_m^{\mathcal{A}})=f^{\mathcal{A}}(s_{m+1}^{\mathcal{A}},\ldots,s_n^{\mathcal{A}})=f^{\mathcal{B}}(c_{m+1}^{\mathcal{B}},\ldots,c_n^{\mathcal{B}})$, and thus $\mathcal{B}$ satisfies $f(c_1,\ldots,c_m) \approx f(c_{m+1},\ldots,c_n)$. We omit the proof for the case of a $D$-flat equation in $S(E)$ because it is similar to the previous cases.

To show the \emph{only}-\emph{if}-\emph{direction}, assume that $s_0\not\approx_{S(E)}^G t_0$. Then, there is an algebra $\mathcal{B}$ over $\Sigma \cup C$ satisfying all equations in $S(E)\cup G$ such that $s_0^{\mathcal{B}}\neq t_0^{\mathcal{B}}$. Let $\mathcal{A}$ be the reduct of $\mathcal{B}$ to $\Sigma\cup C_0$ by forgetting the interpretation of the symbols in the set of new constants $C_1$. Since $s_0^{\mathcal{A}} = s_0^{\mathcal{B}}\neq t_0^{\mathcal{B}} = t_0^{\mathcal{A}}$, we only need to show that $\mathcal{A}$ satisfies the equations in $E \cup G$. Since $\mathcal{B}$ satisfies $G$, it is easy to see that $\mathcal{A}$ satisfies $G$ too. It remains to show that $\mathcal{A}$ satisfies each equation in $E$. More specifically, it suffices to show that $\mathcal{A}$ satisfies each equation after step 1 and step 2 of Phase I from $E$ because $\mathcal{A}$ satisfies $G$. 

Let $E_2$ be the set of equations after step 1 and step 2 of Phase I from $E$. Observe also that for each term $s_k \in Sub(E_2)\setminus C_0$, there is a new constant $c_k \in C_1$. First, consider an equation of the form $a_i \approx b_j \in E_2$, where $a_i$ and $b_j$ are constants. Then, both $a_i$ and $b_j$ are the elements of $C_0$, so the reduct $\mathcal{A}$ obviously satisfies $a_i \approx b_j$, i.e.,\;$a_i^\mathcal{A}=a_i^\mathcal{B}=b_j^\mathcal{B}=b_j^\mathcal{A}$. Next, consider an equation of the form $f(s_1, \ldots, s_m)\approx f(s_{m+1},\ldots,s_n) \in E_2$, where $f\in \Sigma_A$ and $s_1,\ldots, s_n \in Sub(E_2)$. Here, if $s_k\in Sub(E_2)\setminus C_0$, then there is a new constant $c_k \in C_1$ such that $c_k \approx_{S(E)} s_k$. (This can be shown easily by a simple structural induction on term $s_k$ using Definition~\ref{defn:gequations}). Then we have $f(c_1,\ldots,c_m) \approx f(c_{m+1},\ldots, c_n) \in S(E)$ for some $c_1,\ldots, c_n \in C$, where $c_k$ is simply $s_k$ if $s_k \in C_0$. Since $f^{\mathcal{B}}(c_1^{\mathcal{B}},\ldots,c_m^{\mathcal{B}})=f^{\mathcal{B}}(c_{m+1}^{\mathcal{B}},\ldots,c_n^{\mathcal{B}})$, we have $f^{\mathcal{A}}(s_1^{\mathcal{A}},\ldots,s_m^{\mathcal{A}})=f^{\mathcal{B}}(c_1^{\mathcal{B}},\ldots,c_m^{\mathcal{B}})=f^{\mathcal{B}}(c_{m+1}^{\mathcal{B}},\ldots,c_n^{\mathcal{B}})=f^{\mathcal{A}}(s_{m+1}^{\mathcal{A}},\ldots,s_n^{\mathcal{A}})$, and thus $\mathcal{A}$ satisfies $f(s_1,\ldots,s_m) \approx f(s_{m+1},\ldots,s_n)$. We omit the proofs of the remaining cases because they are similar to the previous cases.
\end{proof}

Note that each argument of a term headed by an associative symbol $f$ is a constant in $C$ after Phase II. By a slight abuse of notation, the arguments of a term headed by an associative symbol $f$ are converted into the corresponding string. We define the length of a string $u$ over $C$. If $u\in C^*$, then the \emph{length} of $u$, denoted $|u|$, is defined as: $|\lambda| := 0$, $|c| := 1$ for each $c\in C$, and $|sc| := |s| + 1$ for $s\in C^*$ and $c \in C$, where $\lambda$ denotes the empty string.

\begin{figure}[t]
\hrule
$\newline\newline$
\centering
\begin{tabular}{l}
\noindent \LeftLabel{DEDUCE:\;\;}
\AxiomC{$S\cup \{f(u_1u_2) \approx s, \;f(u_2u_3)\approx t\}$}
\UnaryInfC{$S\cup \{f(u_1t)\approx f(su_3), f(u_1u_2) \approx s, f(u_2u_3)\approx t\}$}
\DisplayProof\\\\
if (i) $f\in \Sigma_{A}$, (ii) $f(u_1u_2) \succ s$, (iii) $f(u_2u_3)\succ t$, and (iv) neither $u_1$ nor $u_2$ nor $u_3$\\ is the empty string, i.e., $|u_1|\neq 0$, $|u_2|\neq 0$, and $|u_3|\neq 0$. \\\\

\noindent \LeftLabel{SIMPLIFY:\;\;}
\AxiomC{$S\cup\{f(u_1u_2u_3)\approx s, \;f(u_2)\approx t\}$}
\UnaryInfC{$S\cup\{\bar{f}(u_1tu_3)\approx s,  \;f(u_2)\approx t\}$}
\DisplayProof\\\\
if (i) $f\in \Sigma_{A}$, (ii) $f(u_2)\succ t$, and (iii) $s\succ t$ if $|u_1u_3|=0$.\\\\

\noindent \LeftLabel{COLLAPSE:\;\;}
\AxiomC{$S\cup\{s[u]\approx t, \;u\approx d\}$}
\UnaryInfC{$S\cup\{s[d] \approx t, \;u\approx d\}$}
\DisplayProof\\\\
if (i)\,$s[u]\succ t$, (ii)\,$u \succ d$, (iii)\,$d \in C$, and\,(iv) $u$ is not headed by a function symbol in $\Sigma_A$.\\\\

\noindent \LeftLabel{COMPOSE:\;\;}
\AxiomC{$S\cup\{t\approx c, \;c\approx d\}$}
\UnaryInfC{$S\cup\{t\approx d, \;c\approx d\}$}
\DisplayProof\\\\
if (i) $t\succ c$, (ii) $c\succ d$, and (iii) $c, d \in C$.\\\\

\noindent \LeftLabel{DELETE:\;\;}
\AxiomC{$S\cup\{s\approx s\}$}
\UnaryInfC{$S$}
\DisplayProof\\\\

Above, all the equations are assumed to be ground and flat, and $S$ is a set of ground \\flat equations. 
In the SIMPLIFY rule, if $f\in \Sigma_A$ and $u$ is a nonempty string over $C$, \\then $\bar{f}(u):= f(u)$ if $|u|\geq 2$, and $\bar{f}(u):=u$, otherwise (i.e.,\;$|u|=1$).
\end{tabular}
$\newline$
\hrule
\caption{The inference System $\mathcal{I}$}
    \label{fig:fig1}
\end{figure}

\begin{defi}\normalfont Let $f\in \Sigma_{A}$ and $u$ be a string over $C$ such that $u:=u_1u_2\cdots u_i$ and $|u|\geq 2$. Then, $f(u)$ denotes the term $f(u_1, \ldots, u_i)$. Each of $f(\lambda u)$, $f(u \lambda)$, and $f(\lambda u \lambda)$ also denotes $f(u)$, where $\lambda$ is the empty string.
\end{defi}

For example, if $f\in \Sigma_{A}$ and $f(a, b, c, d)$, we also write $f(abcd)$ for $f(a, b, c, d)$. Since $f$ is variadic for associatively flat terms with the arity being at least 2, the distinction between two notations is clear from context. In the remainder of this paper, two notations are used interchangeably, and we assume that a total precedence on $C_0$ is always given.

\begin{defi}\label{defn:ordering}\normalfont
Let $W$ be an infinite set of constants $\{c_1,c_2, \ldots\}$ such that $W$ is disjoint from $C_0$, $C_1$ a finite subset chosen from $W$, and $C:=C_0\cup C_1$. We define the order $\succ$ on ground (fully) flat terms in $T(\Sigma, C)$ as follows, assuming that a total precedence on $C_0$ is given:
\begin{enumerate}[(i)]
\item 1 is the minimal element (w.r.t.\;$\succ$),
\item $c_i \succ c_j$  if $i < j$ for all $c_i, c_j \in C_1$,
\item $c \succ c'$ if $c \in C_1$ and $c'\in C_0$,
\item $t \succ c$ if $t$ is any term headed by a function symbol $f$ in $\Sigma$ and $c$ is any constant in $C$,
\item $f(s) \succ f(t)$ if $f\in \Sigma_{A}$ and $s \succ_{L} t$, where $s$ and $t$ are strings over $C$ with $|s|,|t|\geq 2$ and $\succ_{L}$ is the \emph{length-lexicographic ordering} on $C^*$.\footnote{The length-lexicographic ordering $\succ_{L}$ on $C^*$ is defined as follows: $s=s_1\cdots s_i \succ_{L} t_1\cdots t_j =t$ if $i > j$, or they have the same length and $s_1\cdots s_i$ comes before $t_1\cdots t_j$ lexicographically using a precedence on $C$.}
\end{enumerate}

If a $D$-flat (resp.\;an $A$-flat) equation is oriented by $\succ$, then we call it the $D$-flat (resp.\;the $A$-flat) rule. Note that by (ii), $C_1$ is a totally ordered set.
\end{defi}

\indent We see that $\succ$ is well-founded on (associatively flattened) terms in $T(\Sigma, C)$. Note that $\succ$ is a strict partial order on (associatively flattened) terms in $T(\Sigma, C)$ and is a total order on $C$. Yet, it suffices for the inference system $\mathcal{I}$ in Figure~\ref{fig:fig1}. In Figure~\ref{fig:fig1}, DEDUCE is the only expansion rule, and the remaining rules in $\mathcal{I}$ are the contraction rules. (The DEDUCE and SIMPLIFY rules are adapted from finding critical pairs in string rewriting systems (or \emph{semi}-\emph{Thue systems})~\cite{KapurN85,Holt2005}.)

Now, the purpose of Phase III is to perform the ground completion procedure on $S(E)$. The final set of ground completion using Phase III provides a ground convergent rewrite system, where each equation is oriented by $\succ$.\\
 
\noindent {\bf Phase III}: Given $S(E)$ obtained from Phase II, apply the ground completion procedure using the inference rules in Figure~\ref{fig:fig1}. 

\begin{defi}\label{defn:fair}\normalfont (i) We write $S \vdash S'$ to indicate that $S'$ is obtained from $S$ by application of an inference rule in $\mathcal{I}$ (see Figure~\ref{fig:fig1}), where $S$ is a set of equations.\\
(ii) A \emph{derivation} (w.r.t.\;$\mathcal{I}$) is a sequence of states $S_0 \vdash S_1 \vdash\cdots$.\\
(iii)  A derivation $S_0=S(E)\vdash S_1\vdash \cdots$ is \emph{fair} if any rule in $\mathcal{I}$ that is continuously enabled is applied eventually (cf.~\cite{Tiwari2003}).\\
(iv) We denote by $S_{\infty}(E)=\bigcup_i\bigcap_{j\geq i} S_j$ a set of persisting equations obtained by a fair derivation (w.r.t.\;$\mathcal{I}$) starting with $S_0=S(E)$.
\end{defi}

\begin{defi}\normalfont Let $S$ be a set of equations such that each equation is orientable by $\succ$. By $S^\succ$, we denote the rewrite system corresponding to $S$, where each equation $s\approx t \in S$ with $s\succ t$ is oriented into the rewrite rule $s\rightarrow t$.
\end{defi}

It is easy to see that each equation in $S_0=S(E)$ is orientable by $\succ$. Also, each equation in $S_i$ in a derivation starting from $S_0$ is orientable by $\succ$ or a trivial equation because the conclusion of each inference rule in $\mathcal{I}$ is either orientable by $\succ$ or a trivial equation. This means that $S_\infty(E)$ is orientable by $\succ$ (cf.\;a \emph{non}-\emph{failing run}~\cite{Baader1998}).

When rewriting with associatively flat terms (cf.~\cite{Marche1996}), one does not need to introduce \emph{extension rules}~\cite{Dershowitz2001} explicitly. This means that it is not needed to introduce new extension variables for extension rules, which simplifies the inference rules significantly. The notion of overlapping has to be generalized accordingly; two rewrite rules with the same top associative symbol $f$ (at the top position) overlap if they either overlap in the standard way or overlap involving their extension rules. For example, if $f\in \Sigma_A$, then there is a critical overlap between $f(a,b)\rightarrow d$ and $f(b,c)\rightarrow e$ because $f(a,b,c)$ can be written either to $f(d,c)$ or $f(a,e)$, i.e., the critical pair obtained from $f(a,b)\rightarrow d$ and $f(b,c)\rightarrow e$ is $f(a,e)\approx f(d,c)$. Here, extension rules are dealt implicitly without explicitly generating them.

\begin{defi}\label{defn:cpair}\normalfont Let $R$ be the rewrite system obtained from a set of ground flat equations oriented by $\succ$. The \emph{critical pair} (\emph{modulo associativity}) between two rewrite rules in $R$ has one of the following forms:\\
(i) The critical pair obtained from $f(u_1u_2)\rightarrow s \in R$ and $f(u_2u_3)\rightarrow t \in R$ is $f(u_1t)\approx f(su_3)$, where $f\in \Sigma_A$, $|u_1|\neq 0, |u_2|\neq 0$, and $|u_3| \neq 0$.\\
(ii) The critical pair obtained from $f(u_1u_2u_3)\rightarrow s\in R$ and $f(u_2)\rightarrow t\in R$ is $\bar{f}(u_1tu_3)\approx s$, where $f\in \Sigma_{A}$, and $s\succ t$ if $|u_1u_3|= 0$.\\
(iii) The critical pair obtained from $s[u]\rightarrow t\in R$ and $u\rightarrow d\in R$ is $s[d] \approx t$, where $u$ is not headed by a function symbol in $\Sigma_A$ and $d \in C$.\\

By $CP_A(R)$, we denote the set of all critical pairs (modulo associativity) between the rewrite rules in $R$.
\end{defi}

\begin{lem}\label{lem:fair} If a derivation $S_0=S(E)\vdash S_1\vdash \cdots$ is fair, then $CP_A(S_\infty^\succ(E)) \subseteq \bigcup_jS_j$.
\end{lem}
\begin{proof}
Suppose that a derivation $S_0=S(E)\vdash S_1\vdash \cdots$ is fair. Then any rule in $\mathcal{I}$ that is continuously enabled is applied eventually by Definition~\ref{defn:fair}(iii). We consider each critical pair of the form in Definition~\ref{defn:cpair} for the rewrite system $S_\infty^\succ(E)$. If a critical pair in $CP_A(S_\infty^\succ(E))$ was of the form in Definition~\ref{defn:cpair}(i), then the DEDUCE rule was applied. By fairness, this critical pair should be an element of $\bigcup_jS_j$. Similarly, if a critical pair in $CP_A(S_\infty^\succ(E))$ was of the form in Definition~\ref{defn:cpair}(ii), then the SIMPLIFY rule was applied. Again, by fairness, it should be an element of $\bigcup_jS_j$. Finally, if a critical pair in $CP_A(S_\infty^\succ(E))$ was of the form in Definition~\ref{defn:cpair}(iii), then the COLLAPSE rule was applied, so it should also be an element of $\bigcup_jS_j$ by fairness.
\end{proof}

\begin{exa}\label{ex:ex3}\normalfont (Continued from Example~\ref{ex:ex2}) In Example~\ref{ex:ex2} after Phase II, we have $E'=\{h(a)\approx c_1, i(a)\approx c_2, f(aa)\approx 1, f(ac_1)\approx b, f(c_2b)\approx b, i(b)\approx c_3, i(c_1)\approx c_4\}$ and $S(E)=E'\cup I(E')\cup U(C)$, where $C=\{1,a,b,c_1,c_2,c_3,c_4\}$ and $I(E')=\{i(1)\approx 1, i(c_2)\approx a, i(c_3)\approx b,i(c_4)\approx c_1,f(ac_2)\approx 1, f(c_2a)\approx 1,f(bc_3)\approx 1, f(c_3b)\approx 1, f(c_1c_4)\approx 1, f(c_4c_1)\approx 1\}$. Then,\\

\noindent $1': f(bc_4)\approx a$ (DEDUCE by $f(ac_1)\approx b$ and $f(c_1c_4)\approx 1$. Then, SIMPLIFY replacing $f(a1)$ by $a$ using $f(a1) \approx a \in U(C)$.)\\
$2': f(c_3a)\approx c_4$ (DEDUCE by $f(c_3b)\approx 1$ and $1'$. Then, SIMPLIFY replacing $f(1c_4)$ by $c_4$ using $f(1c4) \approx c_4 \in U(C)$.)\\
$3': f(c_2a)\approx f(bc_4)$ (DEDUCE by $f(c_2b)\approx b$ and $1'$.)\\
$4': f(c_2a)\approx a$ (SIMPLIFY $3'$ by $1'$. $3'$ is deleted.)\\
$5': a\approx 1$ (SIMPLIFY $4'$ by $f(c_2a)\approx 1$. $4'$ is deleted.)\\
$6': f(1c_2)\approx 1$ (COLLAPSE $f(ac_2)\approx 1$ by $5'$. $f(ac_2)\approx 1$ is deleted.)\\
$7': c_2\approx 1$ (SIMPLIFY $f(1c_2)\approx c_2$ by $6'$. $f(1c_2)\approx c_2$ is deleted.)\\
$8': f(1c_1)\approx b$ (COLLAPSE $f(ac_1)\approx b$ by $5'$. $f(ac_1)\approx b$ is deleted.)\\
$9': c_1\approx b$ (SIMPLIFY $f(1c_1)\approx c_1$ by $8'$. $f(1c_1)\approx c_1$ is deleted.)\\
$10': h(1)\approx c_1$ (COLLAPSE $h(a)\approx c_1$ by $5'$. $h(a)\approx c_1$ is deleted.)\\
$11': h(1)\approx b$ (COMPOSE $10'$ by $9'$. $10'$ is deleted.)\\
$12': f(c_31)\approx c_4$ (COLLAPSE $2'$ by $5'$. $2'$ is deleted.)\\
$13': c_3\approx c_4$ (SIMPLIFY $f(c_31)\approx c_3$ by $12'$. $f(c_31)\approx c_3$ is deleted.)\\
$\cdots$

\noindent Note that DEDUCE with the equations in $U(C)$ (e.g.,\;$f(b1)\approx b$, $b\in C$) is not necessary. After several steps using the contraction rules, we have $S_\infty(E) =\{h(1)\approx b, c_2\approx 1, a\approx 1, i(1) \approx 1, c_1 \approx b, c_3\approx c_4, i(b)\approx c_4, i(c_4)\approx b, f(bc_4)\approx 1, f(c_4b)\approx 1\}\cup \bar{U}(C)$, where $\bar{U}(C)=\{f(11)\approx 1, f(b1)\approx b, f(1b)\approx b, f(c_41)\approx c_4, f(1c_4)\approx c_4\}$. (Here, 1 is minimal w.r.t.\;$\succ$ and $c_1\succ c_2\succ c_3 \succ c_4\succ b$, where $c_1,c_2,c_3,c_4\in C_1$ and $b\in C_0$ (see Definition~\ref{defn:ordering}).)
\end{exa}

\subsection{Ground Convergent Rewrite Systems for Congruence Closure Modulo Groups}

Recall that $R(G)$ denotes the convergent rewrite system for $G$ modulo $A$. By $gr(R(G))$ we denote $gr(R(G)):=\{l\sigma \rightarrow r\sigma\,|\, l\rightarrow r \in R(G)\;\wedge\;\sigma \text{ is a ground substitution}\}$, and by $gr(A)$ we denote $gr(A):=\{e\sigma\,|\,e \in A \;\wedge\; \sigma \text{ is a ground substitution}\}$. Here, a ground substitution maps variables to (ground) terms in $T(\Sigma,C)$.

\begin{lem}\label{lem:equiv} If $S \vdash S'$, then the congruence relations $\xleftrightarrow{*}_{S\cup gr(A)}$ and $\xleftrightarrow{*}_{S'\cup gr(A)}$ on $T(\Sigma, C)$ are the same.
\end{lem}
\begin{proof} We consider each application of a rule $\tau$ for $S \vdash S'$.

If $\tau$ is DEDUCE, then we need to show that $f(u_1t)\xleftrightarrow{*}_{S\cup gr(A)} f(su_3)$, where $S'-S =\{f(u_1t)\approx f(su_3)\}$. Since $f(u_1u_2u_3)\xleftrightarrow{*}_{S\cup gr(A)} f(su_3)$ and $f(u_1u_2u_3)\xleftrightarrow{*}_{S\cup gr(A)} f(u_1t)$, we have $f(u_1t)\xleftrightarrow{*}_{S\cup gr(A)} f(su_3)$.

If $\tau$ is SIMPLIFY, then we consider two cases. First, suppose that $|u_1u_3|\neq0$. Then, we show that $f(u_1tu_3)\xleftrightarrow{*}_{S\cup gr(A)} s$, where $S'-S =\{f(u_1tu_3)\approx s\}$. Since $f(u_1u_2u_3)\xleftrightarrow{*}_{S\cup gr(A)} s$ and $f(u_1u_2u_3)\xleftrightarrow{*}_{S\cup gr(A)} f(u_1tu_3)$, we have $f(u_1tu_3)\xleftrightarrow{*}_{S\cup gr(A)} s$. Conversely, we show that $f(u_1u_2u_3)\xleftrightarrow{*}_{S'\cup gr(A)} s$, where $S-S'=\{f(u_1u_2u_3)\approx s\}$. Since $f(u_1tu_3)\xleftrightarrow{*}_{S'\cup gr(A)} s$ and $f(u_1u_2u_3)\xleftrightarrow{*}_{S'\cup gr(A)} f(u_1tu_3)$, we have $f(u_1u_2u_3)\xleftrightarrow{*}_{S'\cup gr(A)} s$. Otherwise, suppose that $|u_1u_3|=0$. Then, we need to show that $t\xleftrightarrow{*}_{S\cup gr(A)} s$, where $S'-S =\{t\approx s\}$. Since $f(u_2)\xleftrightarrow{*}_{S\cup gr(A)} s$ and $f(u_2)\xleftrightarrow{*}_{S\cup gr(A)} t$, we have $t\xleftrightarrow{*}_{S\cup gr(A)} s$. Conversely, we show that $f(u_2)\xleftrightarrow{*}_{S'\cup gr(A)} s$, where $S-S'=\{f(u_2)\approx s\}$. Since $t\xleftrightarrow{*}_{S'\cup gr(A)} s$ and $f(u_2)\xleftrightarrow{*}_{S'\cup gr(A)} t$, we have $f(u_2)\xleftrightarrow{*}_{S'\cup gr(A)} s$.

If $\tau$ is COLLAPSE, then we show that $s[d]\xleftrightarrow{*}_{S\cup gr(A)} t$, where $S'-S =\{s[d]\approx t\}$. Since $s[u]\xleftrightarrow{*}_{S\cup gr(A)} t$ and $u\xleftrightarrow{*}_{S\cup gr(A)} d$, we have $s[d]\xleftrightarrow{*}_{S\cup gr(A)} t$.  Conversely, we show that $s[u]\xleftrightarrow{*}_{S'\cup gr(A)} t$, where $S-S'=\{s[u]\approx t\}$. Since $s[d]\xleftrightarrow{*}_{S'\cup gr(A)} t$ and $u\xleftrightarrow{*}_{S'\cup gr(A)} d$, we see that $s[u]\xleftrightarrow{*}_{S'\cup gr(A)} t$.

If $\tau$ is COMPOSE, then we show that $t\xleftrightarrow{*}_{S\cup gr(A)} d$, where $S'-S =\{t\approx d\}$. Since $t\xleftrightarrow{*}_{S\cup gr(A)} c$ and $c\xleftrightarrow{*}_{S\cup gr(A)} d$, it is easy to see that $t\xleftrightarrow{*}_{S\cup gr(A)} d$. Conversely, since $t\xleftrightarrow{*}_{S'\cup gr(A)} d$ and $c\xleftrightarrow{*}_{S'\cup gr(A)} d$, we have $t\xleftrightarrow{*}_{S'\cup gr(A)} c$, where $S-S'=\{t\approx c\}$.

Finally, if $\tau$ is DELETE, then it is immediate that $s\xleftrightarrow{*}_{S'\cup gr(A)} s$.
\end{proof}

\begin{lem}\label{lem:cctheory}If $s_0, t_0 \in T(\Sigma, C_0)$, then $s_0 \approx t_0 \in CC^G(E)$ iff $s_0\approx_{S_\infty(E) \cup gr(R(G))\cup gr(A)} t_0$, where the rewrite system $gr(R(G))$ is viewed as a set of equations.
\end{lem}
\begin{proof}
By Birkhoff's theorem, $CC^G(E)$ coincides with $\approx_E^G$. By Lemma~\ref{lem:conservative}, $s_0\approx_E^G t_0$ iff $s_0\approx_{S(E)}^G t_0$ for all terms $s_0, t_0 \in T(\Sigma, C_0)$. Also, $s_0\approx_{S(E)}^G t_0$ iff $s_0\approx_{S(E) \cup gr(R(G))\cup gr(A)} t_0$ for all terms $s_0, t_0 \in T(\Sigma, C_0)$. 

It remains to show that  $s_0\approx_{S(E) \cup gr(R(G))\cup gr(A)} t_0$ iff $s_0\approx_{S_\infty(E) \cup gr(R(G))\cup gr(A)} t_0$ for all terms $s_0, t_0 \in T(\Sigma, C_0)$, where $S_0=S(E)$. By Lemma~\ref{lem:equiv}, if $S_i \vdash S_{i+1}$, then $\xleftrightarrow{*}_{S_i\cup gr(A)}$ and $\xleftrightarrow{*}_{S_{i+1}\cup gr(A)}$ on $T(\Sigma, C_0)$ are the same (i.e., $u\xleftrightarrow{*}_{S_i\cup gr(A)}v$ iff $u\xleftrightarrow{*}_{S_{i+1}\cup gr(A)}v$ for all $u, v\in T(\Sigma, C_0)$), and thus the conclusion follows.
\end{proof}

\begin{defi}\normalfont\label{defn:cproof} Let $s=s[u]\leftrightarrow s[v]=t$ be a proof step using an equation (rule) $u\approx v \in S\cup gr(A)$. The \emph{complexity} of this proof step is defined as follows:\\
\begin{tabular}{llll}
$(\{s\}, u, t)$ \quad & if $u\approx v \in S$ and $u\succ v$ &\quad\quad\quad\quad $(\{t\}, v, s)$  \quad & if $u\approx v \in S$ and $v\succ u$\\
$(\{s,t\}, \bot, \bot)$ \quad & if $u\approx v\in S$ and $u=v$&\quad\quad\quad\quad $(\{s\},\bot,t)$ \quad & if $u\approx v \in gr(A)$
\end{tabular}
\end{defi}

In Definition~\ref{defn:cproof}, $\bot$ is a new constant symbol. We define the (strict) ordering $\gtrdot$ on $T(\Sigma, C)\cup \{\bot\}$ such that (i) $\bot$ is minimal w.r.t.\;$\gtrdot$, and (ii) $\gtrdot$ is simply $\succ$ on $T(\Sigma, C)$. (As usual, we assume that we use associative flattening for $\gtrdot$ whenever necessary.) Complexities of proof steps are compared lexicographically using the multiset extension of $\gtrdot$ in the first component, and $\gtrdot$ in the second and third component. The \emph{complexity of a proof} is defined by the multiset of the complexities of its proof steps~\cite{Tiwari2003, Bachmair1991}. The \emph{proof ordering}, denoted by $\succ_\mathcal{C}$, is the multiset extension of the ordering on the complexities of proof steps in order to compare the complexities of proofs. (The empty proof~\cite{Bachmair1991} is assumed to be minimal w.r.t.\;$\succ_\mathcal{C}$.) As $\succ$ is well-founded on associatively flattened terms in $T(\Sigma, C)$, $\gtrdot$ is well-founded on associatively flattened terms in $T(\Sigma, C)\cup \{\bot\}$. Since the lexicographic extension and the multiset extension of a well-founded ordering are also well-founded~\cite{Baader1998}, we may infer that $\succ_\mathcal{C}$ is well-founded. 

In the following, if $f\in \Sigma_{A}$ and $u,v,w \in C^*$ such that $u:=u_1\cdots u_i$, $v:=v_1\cdots v_j$, $w:=w_1\cdots w_k$, and $|uw|\neq 0$ and $|v|\geq 2$, then $f(uf(v)w)$ denotes $f(u_1,\ldots,u_i, f(v_1,\ldots,v_j), w_1,\\ \ldots, w_k)$ if $u,w \neq \lambda$, $f(u_1,\ldots,u_i, f(v_1,\ldots,v_j))$ if $u\neq \lambda$ and $w=\lambda$, and  $f(f(v_1,\ldots,v_j), w_1,\\ \ldots, w_k)$ if $u=\lambda$ and $w\neq\lambda$ . For example, if $u=c_1$, $v=c_2c_3$, and $w=\lambda$, then $f(uf(v)w)$ denotes $f(c_1,f(c_2,c_3))$. 

If $s=s[u]\leftrightarrow s[v]=t$ is a proof step using an equation $u\approx v \in S$ (resp.\;$u\approx v \in gr(A)$), then we also write $s\xleftrightarrow{u\approx v}_S t$ (resp.\;$s\xleftrightarrow{u\approx v}_{gr(A)} t$). 

\begin{lem}\label{lem:proofcomplexity} Suppose $S \vdash S'$. Then, for any two terms $s,t\in T(\Sigma, C)$, if $\rho$ is a ground proof in $S \cup gr(A)$ of an equation $s\approx t$, then there is a ground proof $\rho'$ in $S'\cup gr(A)$ of $s\approx t$ such that $\rho\succeq_\mathcal{C} \rho'$.
\end{lem}
\begin{proof} We show that each equation in $S-S'$ has a smaller proof (w.r.t.\;$\succ_\mathcal{C}$) in $S'\cup gr(A)$ by considering each case of $S \vdash S'$ except by the application of the DEDUCE rule. (Note that the DEDUCE rule does not delete any equation, so the conclusion follows immediately.)

(i) SIMPLIFY: We consider two cases. First, suppose that $|u_1u_3|\neq 0$. Then the proof $f(u_1u_2u_3)\xleftrightarrow{f(u_1u_2u_3)\approx s}_{S} s$ is transformed to the proof $f(u_1u_2u_3)\xleftrightarrow{e_1}_{gr(A)} f(u_1f(u_2)u_3)\xleftrightarrow{e_2}_{S'}f(u_1tu_3)\xleftrightarrow{e_3}_{S'} s$, where $e_1=f(u_1u_2u_3)\approx f(u_1f(u_2)u_3)$, $e_2=f(u_2)\approx t$, and $e_3=f(u_1tu_3)\approx s$.  The newer proof is smaller (w.r.t.\;$\succ_\mathcal{C})$ because $f(u_1u_2u_3)\xleftrightarrow{f(u_1u_2u_3)\approx s}_{S} s$ with complexity $(\{f(u_1u_2u_3)\}, f(u_1u_2u_3), s)$ is bigger (w.r.t.\;$\succ_\mathcal{C})$ than the proof step $f(u_1u_2u_3)\xleftrightarrow{e_1}_{gr(A)} f(u_1f(u_2)u_3)$ with complexity $(\{f(u_1u_2u_3)\},\bot, f(u_1f(u_2)u_3))$ in the second component, bigger than the proof step $f(u_1f(u_2)u_3)\xleftrightarrow{e_2}_{S'} f(u_1tu_3)$ with  complexity $(\{f(u_1f(u_2)u_3)\}, f(u_2), f(u_1tu_3))$ in the second component, and bigger than the proof step $f(u_1tu_3)\xleftrightarrow{e_3}_{S'} s$ in the first component. (Here, both $f(u_1u_2u_3)$ and $f(u_1f(u_2)u_3)$ are equal w.r.t.\;$\gtrdot$ because their associatively flat forms are the same.)

Otherwise, suppose that $|u_1u_3|=0$. Then the proof $f(u_2)\xleftrightarrow{f(u_2)\approx s}_{S} s$ is transformed to the proof $f(u_2)\xleftrightarrow{e_1}_{S'} t\xleftrightarrow{e_2}_{S'} s$, where $e_1=f(u_2)\approx t$ and $e_2=t\approx s$. Since $s\gtrdot t$ by the condition of the rule, we see that the newer proof is smaller (w.r.t.\;$\succ_\mathcal{C})$ because $f(u_2)\xleftrightarrow{f(u_2)\approx s}_{S} s$ with complexity $(\{f(u_2)\}, f(u_2), s)$ is bigger (w.r.t.\;$\succ_\mathcal{C})$ than the proof step $f(u_2)\xleftrightarrow{e_1}_{S'} t$ with complexity $(\{f(u_2)\}, f(u_2), t)$ in the last component, and bigger than the proof step $t\xleftrightarrow{e_2}_{S'} s$ with complexity $(\{s\}, s, t)$ in the first component.

(ii) COLLAPSE: The proof $s[u]\xleftrightarrow{s[u]\approx t}_{S} t$ is transformed to the proof $s[u]\xleftrightarrow{e_1}_{S'} s[d] \xleftrightarrow{e_2}_{S'} t$, where $e_1=u\approx d$ and $e_2=s[d]\approx t$. The newer proof is smaller (w.r.t.\;$\succ_\mathcal{C})$ because $s[u]\xleftrightarrow{s[u]\approx t}_{S} t$ with complexity $(\{s[u]\}, s[u], t)$ is bigger (w.r.t.\;$\succ_\mathcal{C})$ than the proof step $s[u]\xleftrightarrow{e_1}_{S'} s[d]$ in the second component and the proof step $s[d] \xleftrightarrow{e_2}_{S'} t$ in the first component.

(iii) COMPOSE: The proof $t\xleftrightarrow{t\approx c}_{S} c$ is transformed to the proof $t\xleftrightarrow{e_1}_{S'} d \xleftrightarrow{e_2}_{S'} c$, where $e_1=t\approx d$ and $e_2=d\approx c$. The newer proof is smaller (w.r.t.\;$\succ_\mathcal{C})$ because $t\xleftrightarrow{t\approx c}_{S} c$ with complexity $(\{t\}, t, c)$ is bigger (w.r.t.\;$\succ_\mathcal{C})$ than the proof step $t\xleftrightarrow{e_1}_{S'} d$ in the last component and the proof step $d \xleftrightarrow{e_2}_{S'} c$ in the first component.

(iv) DELETE: The proof $s\xleftrightarrow{s\approx s}_{S} s$ is transformed to the empty proof, where the proof $s\xleftrightarrow{s\approx s}_{S} s$ with complexity $(\{s,s\},\bot,\bot)$ is bigger (w.r.t.\;$\succ_\mathcal{C})$ than the empty proof.
\end{proof}

A \emph{peak} is a proof of the form $t_1 \leftarrow_{R/A} t \rightarrow_{R/A} t_2$ and a \emph{cliff} is a proof of the form $t_1 \leftrightarrow_{A} t \rightarrow_{R/A} t_2$ or $t_1 \rightarrow_{R/A} t \leftrightarrow_{A} t_2$ for the rewrite relation $\rightarrow_{R/A}$. Note that a cliff is simply replaced by a rewrite step by $R/A$ for associatively flat terms, i.e., $t_1\rightarrow_{R/A} t_2$, so we do not need to consider cliffs on associatively flat terms. Also, recall that $S_{\infty}(E)=\bigcup_i\bigcap_{j\geq i} S_j$ denotes a set of persisting equations obtained by a fair derivation starting from $S_0=S(E)$, and $S_{\infty}^\succ(E)$ denotes the rewrite system obtained from $S_{\infty}(E)$ by orienting each equation in $S_{\infty}(E)$ into the rewrite rule. Note that the ground rewrite system $S_{\infty}^\succ(E)$ can possibly be infinite.

\begin{lem}\label{lem:canonical2}
The ground rewrite system $S_{\infty}^\succ(E)$ is convergent modulo $A$ on $T(\Sigma, C)$.
\end{lem}
\begin{proof}
We omit the proof that $S_{\infty}^\succ(E)/A$ is terminating because it can be directly inferred from the proof of Lemma~\ref{lem:terminating} below by using the ordering $\succ_{CM}$ on the complexity measures $(D,S,W)$ (see Definition~\ref{defn:measure}) or the simpler (lexicographic) complexity measures $(S,W)$ of associatively flat ground terms in $T(\Sigma, C)$. (Note that $\succ$ is not necessarily a reduction order on associatively flat ground terms in $T(\Sigma, C)$.)

Now, we show that no peak of the form $u_1\leftarrow_{S_\infty^\succ(E)/A} \cdot\rightarrow_{S_\infty^\succ(E)} u_2$ for $u_1, u_2 \in T(\Sigma, C)$ occurs in every minimal proof (w.r.t.\;$\succ_\mathcal{C}$) between $u_1$ and $u_2$ in $S_\infty^\succ(E)\cup gr(A)$. (Note that such a minimal proof exists because $\succ_\mathcal{C}$ is well-founded.)

Suppose, towards a contradiction, that such a peak exists. If a peak is a non-overlap\footnote{A simple example of a non-overlap peak is the case $f(c_1,c, d)\leftarrow_{S_\infty^\succ(E)/A} f(a,b,c,d)\rightarrow_{S_\infty^\succ(E)/A} f(a,b,c_2)$ if $f\in \Sigma_A$, $f(a,b)\rightarrow c_1 \in S_\infty^\succ(E)$, and $f(c,d)\rightarrow c_2 \in S_\infty^\succ(E)$.} $t_1 \leftarrow_{S_\infty^\succ(E)/A} t \rightarrow _{S_\infty^\succ(E)/A} t_2$, then it consists of proof steps $t_1 \xleftrightarrow{}_{S_\infty(E)} t_1'\xleftrightarrow{*}_{gr(A)} t \xleftrightarrow{*}_{gr(A)}t_2'\xleftrightarrow{}_{S_\infty(E)} t_2$, where $t_1'\succ t_1$ and $t_2'\succ t_2$. This proof can be transformed to a proof $t_1 \xrightarrow{}_{S_\infty^\succ(E)/A} t' \xleftarrow{}_{S_\infty^\succ(E)/A} t_2$, which consists of proof steps $t_1 \xleftrightarrow{*}_{gr(A)} t_1^{''}\xleftrightarrow{}_{S_\infty(E)} t' \xleftrightarrow{}_{S_\infty(E)}t_2^{''}\xleftrightarrow{*}_{gr(A)} t_2$. We may infer that the newer proof is smaller because $t$ is bigger than $t_1$ and $t_2$, i.e., $t$ is bigger than each term in the newer proof, which is a contradiction.

Now, consider a peak $s_1 \leftarrow_{S_\infty^\succ(E)/A} s \rightarrow _{S_\infty^\succ(E)/A} s_2$ that is a proper overlap. It consists of proof steps $s_1 \xleftrightarrow{}_{S_\infty(E)} s_1'\xleftrightarrow{*}_{gr(A)} s \xleftrightarrow{*}_{gr(A)}s_2'\xleftrightarrow{}_{S_\infty(E)} s_2$, where $s_1'\succ s_1$ and $s_2'\succ s_2$. Then we have $s_1 \xleftrightarrow{*}_{gr(A)}s_1^{''} \xleftrightarrow{}_{CP_A({S_\infty^\succ(E)})}s_2^{''}\xleftrightarrow{*}_{gr(A)} s_2$, where $CP_A({S_\infty^\succ(E)})$ consists of the equations created by the DEDUCE, SIMPLIFY, and COLLAPSE rule applied on the equations in $S_\infty(E)$. Since $CP_A({S_\infty^\succ(E)}) \subseteq \bigcup_j S_j$ by Lemma~\ref{lem:fair}, there is a proof $s_1\xleftrightarrow{*}_{gr(A)}s_1^{''} \xleftrightarrow{}_{S_k}s_2^{''}\xleftrightarrow{*}_{gr(A)} s_2$ for some $k\geq 0$. We name this proof as $\rho$. Observe that this ground proof $\rho$ in $S_k \cup gr(A)$ is strictly smaller than the original peak $s_1 \leftarrow_{S_\infty^\succ(E)/A} s \rightarrow _{S_\infty^\succ(E)/A} s_2$ because $s$ is bigger than each term in $\rho$. Also, there is a ground proof $\rho'$ in $S_\infty(E) \cup gr(A)$ such that $\rho\succeq_\mathcal{C}\rho'$ by Lemma~\ref{lem:proofcomplexity}. Now, $\rho'$ is strictly smaller than the original peak $s_1 \leftarrow_{S_\infty^\succ(E)/A} s \rightarrow _{S_\infty^\succ(E)/A} s_2$, a contradiction.
\end{proof}

We extend the standard definition of a \emph{proper subterm}~\cite{Dershowitz2001} for associatively flat ground terms in $T(\Sigma,C)$. We say that $f(u)$ is a \emph{proper subterm} of $f(vuw)$ if $f\in \Sigma_A$ and $u,v,w \in C^*$, where $|u|\geq 2$ and $|vw|\neq 0$. The following lemma directly follows from the fairness of a derivation involving the SIMPLIFY, COLLAPSE, and COMPOSE rule.
 
\begin{lem}\label{lem:reduced} For each rule $l \rightarrow r$ in $S_{\infty}^\succ(E)$, the right-hand side $r$ is irreducible by $S_{\infty}^\succ(E)/A$ and every proper subterm of $l$ is irreducible by $S_{\infty}^\succ(E)/A$.
\end{lem}

In the following, we define a  complexity measure of an associatively flat ground term so that the measures are compared for each reduction step by $(S_{\infty}^\succ(E) \cup gr(R(G)))/A$. If the associated ordering on these measures is well-founded and each reduction step by $(S_{\infty}^\succ(E) \cup gr(R(G)))/A$ strictly reduces the measure, then $(S_{\infty}^\succ(E) \cup gr(R(G)))/A$ is terminating. First, the size of an associatively flat ground term alone is not a good measure. For example, by applying the rule  $i(f(x_1,x_2)) \rightarrow f(i(x_2), i(x_1)) \in R(G)$, the associatively flat ground term $i(f(a,b))$ with size 4 is rewritten to the associatively flat ground term $f(i(b), i(a))$ with size 5. Also, a reduction step by an $A$-flat rule in $S_{\infty}^\succ(E)$ may preserve the size of a given associatively flat ground term. 

\begin{defi}\label{defn:measure}\normalfont Let $t$ be an associatively flat ground term in $T(\Sigma,C)$.

\noindent (i) The $d$-\emph{measure} of $t$ is the multiset of all pairs $(d, n)$, where $d$ is the depth of each term headed by an occurrence of the inverse symbol $i \in \Sigma_G$ in $t$, and $n$ is the arity of $f\in \Sigma_G$ occurring right below $i$. If the symbol occurring right below $i$ is not $f\in \Sigma_G$, then $n$ in $(d,n)$ is simply 0. If there is no occurrence of $i \in \Sigma_G$ in $t$, then the corresponding multiset is simply empty.

\noindent (ii) The $s$-\emph{measure} of $t$ is the size of $t$.

\noindent (iii) Given a total precedence on the constant symbols in $C$, we define a weight function $w:C \rightarrow \mathbb{N}$ in such a way that for all $c_1, c_2 \in C$, if $c_1 \succ c_2$, then $w(c_1) > w(c_2)$, where $>$ is the usual order on natural numbers. The $w$-\emph{measure} of $t$ is the sum of the weights of all constants occurring in $t$.

\noindent (iv) The \emph{complexity measure} of an associatively flat ground term is the triple $(D, S, W)$, where $D$ is the $d$-measure of it, $S$ is the $s$-measure of it, and $W$ is the $w$-measure of it. The associated ordering of the complexity measures on associatively flat ground terms, denoted by $\succ_{CM}$, is the lexicographic ordering using the lexicographic and multiset extension of $>$ for the first component, and $>$ for the second and third component, where $>$ is the usual order on natural numbers. (Since the lexicographic extension and the multiple extension of a well-founded order is well-founded, $\succ_{CM}$ on associatively flat ground terms is well-founded.)
\end{defi}

\begin{exa}\label{ex:measure}\normalfont (i) Consider $i(h(i(f(a,b,c))))$ for $i,f \in \Sigma_G$. The $d$-measure of it is the multiset $\{(2,3),(4,0)\}$ because the depth of the term headed by $i$ in $i(f(a,b,c))$ is 2 and the arity of $f$ is 3. Meanwhile, the depth of the term headed by the outermost occurrence of $i$ is 4, but the symbol occurring right below it is not $f\in \Sigma_G$, so the corresponding pair is $(4,0)$. Now, consider an $gr(R(G))/A$-reduction step by a ground instance of the rule $i(f(x_1,x_2)) \rightarrow f(i(x_2), i(x_1)) \in R(G)$. For example, $i(h(i(f(a,b,c))))$ is rewritten to $i(h(i(f(b,c)),i(a)))$ by this reduction step. We see that the $d$-measure of $i(h(i(f(b,c)),i(a)))$ is $\{(1,0), (2,2), (4,0)\}$. Note that the $d$-measure $\{(1,0), (2,2), (4,0)\}$ of $i(h(i(f(b,c)),i(a)))$ is smaller than the $d$-measure $\{(2,3),(4,0)\}$ of $i(h(i(f(a,b,c))))$. (Recall that the reduction steps by $gr(R(G))/A$ are always done on associatively flat ground terms.)

\noindent (ii)  Let $C =\{a,b,c_1,c_2\}$ with the precedence $c_1\succ c_2 \succ a \succ b$. Then we may assign the weight of each symbol in such a way that $w(c_1)=4$ , $w(c_2)=3$, $w(a)=2$, and $w(b)=1$. Now, consider an $S_{\infty}^\succ(E)/A$-reduction step by the $A$-flat rule $f(c_1, a)\rightarrow f(c_2,b)$. For example, $h(f(c_1, a))$ is rewritten to $h(f(c_2, b))$ by this reduction step. The $w$-measure of $h(f(c_2, b))$ is smaller than the $w$-measure of $h(f(c_1, a))$ because $w(f(c_1, a)) = 6$ and  $w(f(c_2,b)) = 4$.
\end{exa}

\begin{lem}\label{lem:terminating}
The ground rewrite relation $(S_{\infty}^\succ(E) \cup gr(R(G)))/A$ is terminating on $T(\Sigma, C)$.
\end{lem}
\begin{proof}
We use the complexity measure $(D,S,W)$ of an associatively flat ground term in $T(\Sigma,C)$ (see Definition~\ref{defn:measure}(iv)) and compare these measures w.r.t.\;$\succ_{CM}$ for each type of $(S_{\infty}^\succ(E) \cup gr(R(G)))/A$-reduction step. First, observe that the rules in $R(G)$ are size-reducing except the rule $i(f(x_1,x_2)) \rightarrow f(i(x_2), i(x_1)) \in R(G)$. Let $t_1$ and $t_2$ be associatively flat ground terms in $T(\Sigma,C)$ such that $(S_{\infty}^\succ(E) \cup gr(R(G)))/A$-reduction step exists from $t_1$ to $t_2$. Let \emph{comp}($t_1$) and \emph{comp}($t_2$) denote the complexity measure of $t_1$ and $t_2$, respectively. We consider each type of $(S_{\infty}^\succ(E) \cup gr(R(G)))/A$-reduction step from $t_1$ to $t_2$. If it is a $gr(R(G))/A$-reduction step by a ground instance of the rule $i(f(x_1,x_2)) \rightarrow f(i(x_2), i(x_1)) \in R(G)$, then  \emph{comp}($t_2$) is smaller (w.r.t.\;$\succ_{CM}$) than \emph{comp}($t_1$) in the first component. If it is a $gr(R(G))/A$-reduction step by a ground instance from the other rules in $R(G)$, then  \emph{comp}($t_2$) is smaller (w.r.t.\;$\succ_{CM}$) than \emph{comp}($t_1$) in the second component. (It could be the case that \emph{comp}($t_2$) is smaller (w.r.t.\;$\succ_{CM}$) than \emph{comp}($t_1$) in the other components too.) If it is an $S_{\infty}^\succ(E)/A$-reduction step
by a $D$-flat rule, then \emph{comp}($t_2$) is smaller (w.r.t.\;$\succ_{CM}$) than \emph{comp}($t_1$) in the second component. If it is an $S_{\infty}^\succ(E)/A$-reduction step by an $A$-flat rule or a constant rule (i.e.,\;$c_m \rightarrow c_n$ for some $c_m,c_n\in C$ with $c_m \succ c_n$), then \emph{comp}($t_2$) is smaller (w.r.t.\;$\succ_{CM}$) than \emph{comp}($t_1$) in the third component. 

Since each $(S_{\infty}^\succ(E) \cup gr(R(G)))/A$-reduction step strictly reduces the complexity measure (w.r.t.\;$\succ_{CM}$) and $\succ_{CM}$ is well-founded on the complexity measures of the associatively flat ground terms in $T(\Sigma,C)$ (see Definition~\ref{defn:measure}(iv)), the conclusion follows.
\end{proof}

\begin{thm}\label{thm:canonical}
The ground rewrite system $S_{\infty}^\succ(E) \cup gr(R(G))$ is convergent modulo $A$ on $T(\Sigma, C)$.
\end{thm}
\begin{proof}
Since $(S_{\infty}^\succ(E) \cup gr(R(G)))/A$ is terminating on $T(\Sigma, C)$ by Lemma~\ref{lem:terminating}, we show that $(S_{\infty}^\succ(E) \cup gr(R(G)))/A$ confluent on $T(\Sigma, C)$ by showing that it is locally confluent on $T(\Sigma, C)$. In this proof, we omit the joinability of the pair obtained by rewriting within the substitution part (cf. \emph{variable overlap}) involving the rules in $gr(R(G))$ from $R(G)$, which is discussed similarly in the literature~\cite{Baader1998,Dershowitz2001}. We also omit the case of non-overlaps involving the rules in $gr(R(G))$ and $S_{\infty}^\succ(E)$, which is also discussed similarly in the literature~\cite{Baader1998,Dershowitz2001}.
 
Now, it suffices to consider the (critical) overlaps between the rules in $gr(R(G))$ and $S_{\infty}^\succ(E)$. Recall that for each rule $l \rightarrow r$ in $S_{\infty}^\succ(E)$, the right-hand side $r$ is irreducible and every proper subterm of $l$ is irreducible by $S_{\infty}^\succ(E)/A$ by Lemma~\ref{lem:reduced}.

Suppose that there is an overlap between a rule $i(i(a_i))\rightarrow a_i$ in $gr(R(G))$ and a rule $i(a_i)\rightarrow c_{a_i}$ in $S_{\infty}^\succ(E)$ with the critical pair $i(c_{a_i}) \approx a_i$. Since $i(c_{a_i})\rightarrow a_i$ in $S_{\infty}^\succ(E)$ by construction (see Phase II), $i(c_{a_i})$ and $a_i$ are joinable by the rule in $S_{\infty}^\succ(E)$.

Suppose that there is an overlap between a rule $f(c, i(c))\rightarrow 1$ (resp.\;$f(i(c),c)\rightarrow 1$) in $gr(R(G))$ and a rule $i(c)\rightarrow d$ in $S_{\infty}^\succ(E)$ with the critical pair $f(c,d)\approx 1$ (resp.\;$f(d,c)\approx 1$). Since $i(c)\rightarrow d$ in $S_{\infty}^\succ(E)$, we also have $f(c,d)\rightarrow 1$ and $f(d,c)\rightarrow 1$ in $S_{\infty}^\succ(E)$ (see Phase II), and thus $f(c,d)$ (resp.\;$f(d,c)$) and $1$ are joinable by the rule in $S_{\infty}^\succ(E)$.\footnote{If one also considers the overlapping involving \emph{extension}~\cite{Dershowitz2001} rules, it is easy to see that the case of overlapping between the ground instances of the extension of $f(c, i(c))\rightarrow 1$ (resp.\;$f(i(c),c)\rightarrow 1$) in $gr(R(G))$ and $i(c)\rightarrow d$ in $S_{\infty}^\succ(E)$, along with the case of overlapping between the ground instances of the extension of $f(c, 1)\rightarrow c$ (resp.\;$f(1,c)\rightarrow 1$) in $gr(R(G))$ and $f(c,1)\rightarrow c$ (resp.\;$f(1,c)\rightarrow c$) in $S_{\infty}^\succ(E)$ are both joinable.}

Suppose that there is an overlap between a rule $f(u_1,\ldots, u_m, 1)\rightarrow f(u_1,\ldots, u_m)$ (resp.\;$f(1, u_1,\ldots, u_m)\rightarrow f(u_1,\ldots, u_m)$) in $gr(R(G))$ and a rule $f(u_m,1)\rightarrow u_m$ (resp.\;$f(1,u_1)$ $\rightarrow$ $u_1$) in $S_{\infty}^\succ(E)$. Then, the critical pair $f(u_1,\ldots, u_m) \approx f(u_1,\ldots, u_m)$ is trivially joinable.

Suppose that there is an overlap between a rule of the form $i(f(u_1,\ldots,u_n))$ $\rightarrow$ $f(i(\bar{f}(u_{m+1},\ldots, u_n))$, $i(\bar{f}(u_1,\ldots,u_m)))$ in $gr(R(G))$ and a rule $f(u_1,\ldots,u_n) \rightarrow u_m$ in $S_{\infty}^\succ(E)$. (Here, $\bar{f}(u_k,\ldots, u_1)$ is $f(u_k, \ldots, u_1)$ if $k\geq 2$, and $u_1$, otherwise (i.e., $k=1$).) We show that $f(i(\bar{f}(u_{m+1},\ldots, u_n))$, $i(\bar{f}(u_1,\ldots,u_m)))$ and $i(u_m)$ are joinable by $\rightarrow_{S_{\infty}^\succ(E)\cup gr(R(G))/A}$.\,As $f(i(\bar{f}(u_{m+1},\ldots, u_n))$, $i(\bar{f}(u_1,\ldots,u_m)))$ $ \xrightarrow{*}_{gr(R(G))/A} f(i(u_n),i(u_{n-1}),\\\ldots, i(u_1))$, we show that $f(i(u_n), i(u_{n-1}),\ldots, i(u_1))$ and $i(u_m)$  
are joinable by $\rightarrow_{S_{\infty}^\succ(E)/A}$. By construction, we have some rules $i(u_1)\rightarrow c_{u_1}, i(u_2)\rightarrow c_{u_2}, \ldots, i(u_n)\rightarrow c_{u_n}$ and $i(u_m)\rightarrow c_{u_m}$ in $S_{\infty}^\succ(E)$, where $u_1,\ldots, u_n, u_m$ and $c_{u_1},\ldots, c_{u_n},c_{u_m}$ may not be distinct. We also have the rules $f(u_1, c_{u_1})\rightarrow 1,f(c_{u_1},u_1)\rightarrow 1, \ldots, f(u_n, c_{u_n})\rightarrow 1, f(c_{u_n}, u_n)\rightarrow 1, f(u_m, c_{u_m})\rightarrow 1, f(c_{u_m}, u_m)\rightarrow 1$ in $S_{\infty}^\succ(E)$. Now, it suffices to show that $f(c_{u_n},\ldots, c_{u_1})$ and $c_{u_m}$ are joinable by $\rightarrow_{S_{\infty}^\succ(E)/A}$. Since $f(u_1,\ldots,u_n) \rightarrow u_m$ in $S_{\infty}^\succ(E)$, we have $f(c_{u_m},u_1,\ldots,u_n, c_{u_n},\ldots, c_{u_1})\\ \rightarrow_{S_{\infty}^\succ(E)/A} f(c_{u_m},u_m, c_{u_n},\ldots, c_{u_1})$ by applying the same context to the rule $f(u_1,\ldots,u_n) \rightarrow u_m$ in $S_{\infty}^\succ(E)$. This means that $f(c_{u_m},u_1,\ldots,u_n, c_{u_n},\ldots, c_{u_1})$ and $f(c_{u_m},u_m, c_{u_n},\ldots, c_{u_1})$ are joinable by $\rightarrow_{S_{\infty}^\succ(E)/A}$. Since $f(c_{u_m},u_1,\ldots,u_n, c_{u_n},\ldots, c_{u_1})\xrightarrow{*}_{S_{\infty}^\succ(E)/A} c_{u_m}$ and $f(c_{u_m},\\u_m, c_{u_n},\ldots, c_{u_1})\xrightarrow{*}_{S_{\infty}^\succ(E)/A} f(c_{u_n},\ldots, c_{u_1})$, by Lemma~\ref{lem:canonical2}, we infer that $f(c_{u_n},\ldots, c_{u_1})$ and $c_{u_m}$ are joinable by $\rightarrow_{S_{\infty}^\succ(E)/A}$.

Finally, suppose that there is an overlap between a rule of the form $i(f(u_1,\ldots,u_n)){\rightarrow}\\f(i(\bar{f}(u_{m+1},\ldots, u_n))$, $i(\bar{f}(u_1,\ldots,u_m)))$ in $gr(R(G))$ and a rule $f(u_1,\ldots,u_n) {\rightarrow} f(v_1,\ldots, v_m)$ in $S_{\infty}^\succ(E)$. (Here, $\bar{f}(u_k,\ldots, u_1)$ is $f(u_k, \ldots, u_1)$ if $k\geq 2$, and $u_1$, otherwise.)  We show that $f(i(\bar{f}(u_{m+1},\ldots, u_n))$, $i(\bar{f}(u_1,\ldots,u_m)))$ and $i(f(v_1,\ldots,v_m))$ are indeed joinable by $\rightarrow_{S_{\infty}^\succ(E) \cup gr(R(G))/A}$. Now, since $f(i(\bar{f}(u_{m+1},\ldots, u_n))$, $i(\bar{f}(u_1,\ldots,u_m)))$  $\xrightarrow{*}_{gr(R(G))/A} f(i(u_n),\\ i(u_{n-1}),\ldots, i(u_1))$, and $i(f(v_1,\ldots,v_m)) \xrightarrow{*}_{gr(R(G))/A} f(i(v_m),\ldots, i(v_1))$, we then show that \\$f(i(u_n), i(u_{n-1}),\ldots, i(u_1))$ and $f(i(v_m),\ldots, i(v_1))$ are joinable by $\rightarrow_{S_{\infty}^\succ(E)/A}$. By construction, we have some rules $i(u_1)\rightarrow c_{u_1}, i(u_2)\rightarrow c_{u_2}, \ldots, i(u_n)\rightarrow c_{u_n}$, $i(v_1)\rightarrow c_{v_1}, \ldots, i(v_m)\rightarrow c_{v_m}$ in $S_{\infty}^\succ(E)$, where $u_1,\ldots, u_n, v_1,\ldots,v_m, c_{u_1},\ldots, c_{u_n}, c_{v_1},\ldots, c_{v_m}$ may not be distinct.

\noindent It suffices to show that $f(c_{u_n},\ldots, c_{u_1})$ and $f(c_{v_m}, \ldots, c_{v_1})$ are joinable by $\rightarrow_{S_{\infty}^\succ(E)/A}$. Since $f(u_1,\ldots,u_n) \rightarrow f(v_1,\ldots, v_m) \in S_{\infty}^\succ(E)$, we see that $f(c_{v_m},\ldots, c_{v_1}, u_1,\ldots,u_n, c_{u_n},\ldots,c_{u_1})\\\rightarrow_{S_{\infty}^\succ(E)/A} f(c_{v_m},\ldots, c_{v_1}, v_1,\ldots, v_m, c_{u_n},\ldots, c_{u_1})$ by applying the same context to the rule $f(u_1,\ldots,u_n) \rightarrow f(v_1,\ldots, v_m)$ in $S_{\infty}^\succ(E)$, so $f(c_{v_m},\ldots, c_{v_1}, u_1,\ldots,u_n, c_{u_n},\ldots, c_{u_1})$ and $f(c_{v_m},\ldots, c_{v_1}, v_1,\ldots, v_m, c_{u_n},\ldots, c_{u_1})$ are joinable by $\rightarrow_{S_{\infty}^\succ(E)/A}$. Since $f(c_{v_m},\ldots, c_{v_1}, u_1,\\\ldots,u_n, c_{u_n},\ldots, c_{u_1}) \xrightarrow{*}_{S_{\infty}^\succ(E)/A} f(c_{v_m},\ldots, c_{v_1})$ and $f(c_{v_m},\ldots, c_{v_1}, v_1,\ldots, v_m, c_{u_n},\ldots, c_{u_1}) \\\xrightarrow{*}_{S_{\infty}^\succ(E)/A} f(c_{u_n},\ldots, c_{u_1})$, by Lemma~\ref{lem:canonical2}, we infer that $f(c_{u_n},\ldots, c_{u_1})$ and $f(c_{v_m},\ldots, c_{v_1})$ are joinable by $\rightarrow_{S_{\infty}^\succ(E)/A}$.
\end{proof}

The following theorem says that given $s_0 \approx t_0$, where $s_0, t_0 \in T(\Sigma, C_0)$, the membership problem for $CC^G(E)$, written $s_0\stackrel{?}{\approx}t_0 \in CC^G(E)$, is reduced to checking whether $s_0$ and $t_0$ have the same normal form w.r.t.\;$(S_\infty^\succ(E) \cup gr(R(G)))/A$.

\begin{thm}\label{thm:ccequiv} If $s_0, t_0 \in T(\Sigma, C_0)$, then $s_0 \approx t_0 \in CC^G(E)$ iff $s_0$ and $t_0$ have the same normal form w.r.t.\;$(S_\infty^\succ(E) \cup gr(R(G)))/A$.
\end{thm}
\begin{proof}
Assume that $s_0, t_0 \in T(\Sigma, C_0)$. If $s_0\approx t_0 \in CC^G(E)$, then $s_0 \approx_{S_\infty(E) \cup gr(R(G))\cup gr(A)} t_0$ by Lemma~\ref{lem:cctheory}. Since $S_{\infty}^\succ(E) \cup gr(R(G))$ is convergent modulo $A$ on $T(\Sigma, C)$ by Theorem~\ref{thm:canonical}, $s_0$ and $t_0$ have the same normal form w.r.t.\;$(S_\infty^\succ(E) \cup gr(R(G))/A$.

Conversely, if $s_0$ and $t_0$ have the same normal form w.r.t.\;$(S_\infty^\succ(E) \cup gr(R(G))/A$, then we have  $s_0 \approx_{S_\infty(E) \cup gr(R(G)) \cup gr(A)} t_0$, and thus $s_0\approx t_0 \in CC^G(E)$  by Lemma~\ref{lem:cctheory}.
\end{proof}

\begin{rem} Note that $s_0, t_0 \in T(\Sigma, C_0)$ with $s_0 \approx t_0 \in CC^G(E)$ in Theorem~\ref{thm:ccequiv} may have the normal forms w.r.t.\;$(S_\infty^\succ(E) \cup gr(R(G)))/A$ in $T(\Sigma, C_1)$, but they are the same by Theorem~\ref{thm:canonical} because $T(\Sigma, C_0) \subseteq T(\Sigma, C)$. (Recall that $C := C_0 \cup C_1$.)
\end{rem}

\begin{rem} Recall that $gr(R(G))$ in Theorem~\ref{thm:ccequiv} is defined as $gr(R(G)):=\{l\sigma \rightarrow r\sigma\,|\, l\rightarrow r \in R(G)\,\wedge\,\sigma \text{ is a ground substitution}\}$. (Here, a ground substitution maps variables to (ground) terms in $T(\Sigma,C)$.) This means that $gr(R(G))$ can be infinite. However, instead of using the infinite ground rewrite system $gr(R(G))$ directly, we can still use the finite rewrite system $R(G)$ for rewriting steps w.r.t.\;$gr(R(G))/A$ by using $A$-matching on (ground) terms in $T(\Sigma,C)$.
\end{rem}

\begin{cor}\label{cor:wordproblem} Given a finite set of ground equations $E\subseteq T(\Sigma, C_0) \times T(\Sigma, C_0)$, if $S_\infty(E)$ is finite,  then we can decide for any $s_0, t_0 \in T(\Sigma, C_0)$ whether $s_0 \approx_E^G t_0$ holds or not.
\end{cor}

\begin{proof}
By Birkhoff's theorem, $CC^G(E)$ coincides with $\approx_E^G$. By Theorem~\ref{thm:ccequiv}, we can decide whether $s_0 \approx_E^G t_0$ using the normal forms of $s_0$ and $t_0$ w.r.t.\;$(S_\infty^\succ(E) \cup gr(R(G)))/A$, i.e., $s_0$ and $t_0$ have the same normal form w.r.t.\;$(S_\infty^\succ(E) \cup gr(R(G)))/A$ iff $s_0 \approx_E^G t_0$.
\end{proof}

\begin{exa}\label{ex:ex4}\normalfont (Continued from Example~\ref{ex:ex3}) Consider the word problem of deciding whether $i(i(f(h(a), f(i(b), a))))\approx_E^G 1$ holds or not using $S_\infty(E)$ in Example~\ref{ex:ex3}. Recall that two notations $f(u)$ and $f(u_1,\ldots, u_n)$ are used interchangeably if $f\in \Sigma_A$ and $u$ is a string over $C$ such that $u:=u_1,\ldots, u_n$ and $|u|\geq 2$. Then, $S_\infty^\succ(E) =\{h(1)\rightarrow b, c_2\rightarrow 1, a\rightarrow 1, i(1) \rightarrow 1, c_1 \rightarrow b, c_3\rightarrow c_4, i(b)\rightarrow c_4, i(c_4)\rightarrow b, f(b, c_4)\rightarrow 1, f(c_4, b)\rightarrow 1\}\cup \bar{U}^\succ (C)$, where $\bar{U}^\succ (C)=\{f(1,1)\rightarrow 1, f(b,1)\rightarrow b, f(1,b)\rightarrow b, f(c_4,1)\rightarrow c_4, f(1,c_4)\rightarrow c_4\}$. First, $i(i(f(h(a), f(i(b), a))))$ is associatively flattened to $i(i(f(h(a), i(b), a)))$. 

Let $SG := S_\infty^\succ(E) \cup gr(R(G))$. Then, $i(i(f(h(a), i(b), a)))\rightarrow_{SG/A} f(h(a), i(b), a)\xrightarrow{*}_{SG/A}f(h(1), i(b), 1)\rightarrow_{SG/A} f(b, i(b), 1)\rightarrow_{SG/A} f(b, c_4, 1)\rightarrow_{SG/A} f(1, 1)\rightarrow_{SG/A}1$. Now, we see that $i(i(f(h(a), f(i(b), a))))\approx_E^G 1$ holds and  $i(i(f(h(a), f(i(b), a))))\approx 1 \in CC^G(E)$ by Theorem~\ref{thm:ccequiv}.
\end{exa}

\begin{defi}\normalfont Given $S_0=S(E)$, let $R(E)$ be the set of $D$-flat and $A$-flat equations containing\;$f\in \Sigma_G$ in $S(E)-U(C)$ (see Phase II). Let $C(R)$ be the set of constant symbols appearing in $R(E)$ except 1. We say that ${<}C(R)\,|\,R(E){>}$ is the \emph{monoid presentation for} $S(E)$.
\end{defi}

\begin{exa}\normalfont \label{ex:ex5}
Let $E=\{f(a,a,a)\approx 1,f(h(a), h(a))\approx 1, f(a,h(a), a,h(a))\approx 1\}$ and $f\in \Sigma_G$. Then $S(E)=\{f(aaa)\approx 1, h(a)\approx c_1, f(c_1c_1)\approx 1, f(ac_1ac_1)\approx 1, i(1)\approx 1, i(a)\approx c_2, i(c_2)\approx a, i(c_1)\approx c_3, i(c_3)\approx c_1, f(ac_2)\approx 1, f(c_2a)\approx 1,f(c_1c_3)\approx 1, f(c_3c_1)\approx 1\} \cup U(C)$, where $C=\{1, a, c_1, c_2, c_3\}$. Now, we have the monoid presentation ${<}C(R)\,|\,R(E){>}$ for $S(E)$, where $C(R)=\{a, c_1, c_2, c_3\}$ and $R(E)=\{f(ac_2)\approx 1, f(c_2a)\approx 1,f(c_1c_3)\approx 1, f(c_3c_1)\approx 1, f(aaa)\approx 1, f(c_1c_1)\approx 1, f(ac_1ac_1)\approx 1\}$. This monoid presentation for $S(E)$ corresponds to a monoid presentation of the dihedral group of order 6~\cite{Holt2005} ${<} \alpha, \beta, \alpha^{-1}, \beta^{-1}, \,|\, \alpha\alpha^{-1}\approx 1, \alpha^{-1}\alpha\approx 1, \beta\beta^{-1}\approx 1, \beta^{-1}\beta \approx 1, \alpha^3 \approx 1, \beta^2 \approx 1, \alpha\beta\alpha\beta \approx 1 {>}$ by renaming the symbols $a$ to $\alpha$, $c_1$ to $\beta$, $c_2$ to $\alpha^{-1}$, and $c_3$ to $\beta^{-1}$.
\end{exa}

It is known that the Knuth-Bendix completion terminates for finite groups using their monoid presentations~\cite{Holt2005} with their associated length-lexicographic ordering. Similarly, if the monoid presentation for $S_0=S(E)$ is a monoid presentation of a finite group, then $S_\infty(E)$ is finite (see Proposition~\ref{prop:monoid_presentation}), providing a decision procedure for the word problem for $E$ w.r.t.\;$G$ by Corollary~\ref{cor:wordproblem}.

\begin{exa}\normalfont\label{ex:ex6} (Continued from Example~\ref{ex:ex5}) Given $S(E)=\{f(aaa)\approx 1, h(a)\approx c_1, f(c_1c_1)\approx 1, f(ac_1ac_1)\approx 1, i(1)\approx 1, i(a)\approx c_2, i(c_2)\approx a, i(c_1)\approx c_3, i(c_3)\approx c_1, f(ac_2)\approx 1, f(c_2a)\approx 1,f(c_1c_3)\approx 1, f(c_3c_1)\approx 1\} \cup U(C)$ and $C=\{1, a, c_1, c_2, c_3\}$ in Example~\ref{ex:ex5}, let $i\succ h \succ f \succ c_1\succ c_2 \succ c_3 \succ a \succ 1$ (see Definition~\ref{defn:ordering}). In the following, we implicitly apply SIMPLIFY using $U(C)$ whenever applicable (cf.\;Example~\ref{ex:ex3}):\\

\noindent $1': c_1 \approx c_3\;($DEDUCE by $f(c_1c_1)\approx 1$ and $f(c_1c_3)\approx 1)$\\
$($By $1'$, each $c_1$ occurring in $S(E)$ is replaced by $c_3$ using COLLAPSE and COMPOSE.$)$\\
$\cdots$\\
$2': f(aa)\approx c_2$ $($DEDUCE by $f(aaa)\approx 1$ and $f(ac_2)\approx 1)$\\
$3': f(c_2c_2)\approx a$ $($DEDUCE by $2'$ and $f(ac_2)\approx 1)$\\
$4': f(ac_3a) \approx c_3$ $($DEDUCE by $f(ac_3ac_3)\approx 1$ and $f(c_3c_3)\approx 1)$\\
$5': f(c_2c_3) \approx f(c_3a)$ $($DEDUCE by $4'$ and $f(c_2a)\approx 1)$\\
$6': f(c_3c_2) \approx f(ac_3)$ $($DEDUCE by $4'$ and $f(ac_2)\approx 1)$\\
$7': f(c_3ac_3) \approx c_2$ $($DEDUCE by $6'$ and $f(c_3c_3)\approx 1)$\\
$\cdots$

\noindent After several steps using the contraction rules, we have the finite set $S_\infty(E)=\{h(a)\approx c_3, f(c_3c_3)\approx 1, i(1)\approx 1, i(a)\approx c_2, i(c_2)\approx a, i(c_3)\approx c_3, f(ac_2)\approx 1, f(c_2a)\approx 1\}\cup \bar{U}(C) \cup \{1',2',3',4',5',6',7'\}$, where $\bar{U}(C)$ is obtained from $U(C)$ by rewriting each occurrence of $c_1$ in $U(C)$ to $c_3$. Now, we have the finite rewrite system $S_\infty^\succ(E)$, which is obtained from $S_\infty(E)$ by orienting each equation in $S_\infty(E)$ into the rewrite rule. Given $S_0=S(E)$, the output depends on a given precedence on symbols, but the completion procedure necessarily terminates (see Proposition~\ref{prop:monoid_presentation}).
\end{exa}

\begin{prop}\label{prop:monoid_presentation} Given $S_0=S(E)$, if the monoid presentation ${<}C(R)\,|\,R(E){>}$ for $S(E)$ is a monoid presentation of a finite group, then $S_\infty(E)$ is finite.
\end{prop}
\begin{proof}
Let $X=C(R)$ and $f\in \Sigma_G$. If the monoid presentation ${<}X\,|\,R(E){>}$ for $S(E)$ is a monoid presentation of a finite group, then there are only finitely many $\xleftrightarrow{*}_{R(E)\cup gr(A)}$ equivalence classes on $\{f(u)\,|\, u\in X^*\text{ and } |u|\geq 2\}$ (see Chapter 12 in~\cite{Holt2005}).

Now, if $X$ is not the same as $C$, then let $\bar{X}$ be the set obtained from $X$ by adding each constant $c'\in C-X$ satisfying $c'\xleftrightarrow{*}_{S(E)}c$ for some  $c \in X$ to $X$. (If either $X=C$ or $X\neq C$ but such $c'$ does not exist, then we simply let $\bar{X}:=X$.) Since there are finitely many $\xleftrightarrow{*}_{R(E)\cup gr(A)}$ equivalence classes on $\{f(u)\,|\, u\in X^*\text{ and } |u|\geq 2\}$, we may infer that there are also finitely many $\xleftrightarrow{*}_{S(E)\cup gr(A)}$ equivalence classes on $\{f(u)\,|\, u\in \bar{X}^*\text{ and } |u|\geq 2\}$.

 By Lemma~\ref{lem:equiv}, we have $\xleftrightarrow{*}_{S(E)\cup gr(A)}= \xleftrightarrow{*}_{S_\infty(E)\cup gr(A)}$, so there are only finitely many $\xleftrightarrow{*}_{S_\infty(E)\cup gr(A)}$ equivalence classes on $\{f(u)\,|\, u\in \bar{X}^*\text{ and } |u|\geq 2\}$. Since $S_\infty^\succ(E)$ is convergent modulo $A$ by Lemma~\ref{lem:canonical2}, this means that there are only finitely many $S_\infty^\succ(E)/A$-normal forms of the terms in $\{f(u)\,|\, u\in \bar{X}^*\text{ and } |u|\geq 2\}$.

Suppose that $S_\infty(E)$ is an infinite set. Since non-$\Sigma_A$ symbols have fixed arities and $\Sigma\cup C$ is finite, $S_\infty(E)$ must contain infinitely many equations with top symbol $f$. As there are only finitely many $S_\infty^\succ(E)/A$-normal forms of the terms in $\{f(u)\,|\, u\in \bar{X}^*\text{ and } |u|\geq 2\}$, there is some rule $f(u_1\cdots u_n)\rightarrow t \in S_\infty^\succ(E)$ such that a proper subterm of $f(u_1\cdots u_n)$ is reducible by $S_\infty^\succ(E)/A$. This is not possible by Lemma~\ref{lem:reduced}, a contradiction.
\end{proof}

\section{Congruence closure modulo semigroups, monoids, and the multiple disjoint sets of group axioms}
Given a finite set of ground equations $E$ between terms $s,t \in T(\Sigma, C_0)$, by congruence closure modulo semigroups (resp.\;monoids), we mean congruence closure modulo the associativity (resp.\;monoid) axioms for $E$. In this section, we first discuss congruence closure modulo semigroups by considering multiple associative symbols. Then we discuss congruence modulo monoids by considering multiple associative symbols but only one of them is an interpreted symbol for the monoid axioms. We also discuss congruence closure modulo the multiple disjoint sets of group axioms. All these approaches only differ by Phases I and II in Section~\ref{sec:cc}, but they use the same Phase III in Section~\ref{sec:cc}. This means that one may use the same completion procedure for constructing congruence closure modulo the semigroup, monoid, and the multiple disjoint sets of group axioms, respectively.

\subsection{Congruence closure modulo semigroups}
Unlike constructing congruence closure modulo the group axioms, constructing congruence closure modulo the semigroup axioms does not need to add certain ground flat equations entailed by the group axioms, so Phase II in Section~\ref{sec:cc} is not necessary. Also, normalizing each term using the rewriting by $R(G)/A$ is not necessary, either. Now, Phase I for constructing congruence modulo semigroups is the same as Phases I in Section~\ref{sec:cc} without step 2. The output of Phase I is denoted by $S(E)$, where all equations in $S(E)$ are constant, $D$-flat, or $A$-flat equations. Then one may apply Phase III in Section~\ref{sec:cc} directly using the same inference rules in Figure~\ref{fig:fig1} with $S_0 = S(E)$ and have the same Lemmas~\ref{lem:equiv} and~\ref{lem:canonical2}. Now, we have the following results adapted from Lemma~\ref{lem:cctheory}, Theorem~\ref{thm:ccequiv} and Corollary~\ref{cor:wordproblem}, respectively.

\begin{lem}\label{lem:cctheory_assoc}
If $s_0, t_0 \in T(\Sigma, C_0)$, then $s_0 \approx t_0 \in CC^A(E)$ iff $s_0\approx_{S_\infty(E) \cup gr(A)} t_0$.
\end{lem}
\begin{proof}
By Birkhoff's theorem, $CC^A(E)$ coincides with $\approx_E^A$. Adapted from Lemma~\ref{lem:conservative}, we may infer that $s_0\approx_E^A t_0$ iff $s_0\approx_{S(E)}^A t_0$ for all terms $s_0, t_0 \in T(\Sigma, C_0)$. Also, $s_0\approx_{S(E)}^A t_0$ iff $s_0\approx_{S(E) \cup gr(A)} t_0$ for all terms $s_0, t_0 \in T(\Sigma, C_0)$. 

It remains to show that  $s_0\approx_{S(E) \cup gr(A)} t_0$ iff $s_0\approx_{S_\infty(E)\cup gr(A)} t_0$ for all terms $s_0, t_0 \in T(\Sigma, C_0)$, where $S_0=S(E)$. By Lemma~\ref{lem:equiv}, if $S_i \vdash S_{i+1}$, then $\xleftrightarrow{*}_{S_i\cup gr(A)}$ and $\xleftrightarrow{*}_{S_{i+1}\cup gr(A)}$ on $T(\Sigma, C_0)$ are the same, and thus the conclusion follows.
\end{proof}

\begin{lem}\label{lem:ccequiv_semigroups} If $s_0, t_0 \in T(\Sigma, C_0)$, then $s_0 \approx t_0 \in CC^A(E)$ iff $s_0$ and $t_0$ have the same normal form w.r.t.\;$S_\infty^\succ(E)/A$.
\end{lem}
\begin{proof}
Assume that $s_0, t_0 \in T(\Sigma, C_0)$. If $s_0\approx t_0 \in CC^A(E)$, then $s_0 \approx_{S_\infty(E)\cup gr(A)} t_0$ by Lemma~\ref{lem:cctheory_assoc}. Since $S_{\infty}^\succ(E)$ is convergent modulo $A$ by Lemma~\ref{lem:canonical2}, $s_0$ and $t_0$ have the same normal form w.r.t.\;$S_\infty^\succ(E)/A$.

Conversely, if $s_0$ and $t_0$ have the same normal form w.r.t.\;$S_\infty^\succ(E)/A$, then we have  $s_0 \approx_{S_\infty(E)\cup gr(A)} t_0$, and thus $s_0\approx t_0 \in CC^A(E)$  by Lemma~\ref{lem:cctheory_assoc}.
\end{proof}

\begin{cor}\label{cor:wordproblem_semigroups}Given a finite set of ground equations $E\subseteq T(\Sigma, C_0) \times T(\Sigma, C_0)$, if $S_\infty(E)$ is finite,  then we can decide for any $s_0, t_0 \in T(\Sigma, C_0)$ whether $s_0 \approx_E^A t_0$ holds or not.
\end{cor}

In the remainder of this section, the examples are only concerned with a set of associatively flattened, ground fully flat equations $E$ for simplicity. (See Example~\ref{ex:ex1} in Section~\ref{sec:cc} for an example of the associatively flattening and the fully flattening step, respectively.)

\begin{exa}\label{ex:ex7}\normalfont
Let $E=\{f(a,b)\approx a, f(b,c)\approx b, c\approx d\}$ with $f\in \Sigma_A$ and $a\succ b \succ c \succ d$. Each term in $E$ is already associatively flattened. Also, the equation in $E$ is already a ground fully flat equation, so Phase I is not needed. Phase II is not needed either for $CC^A(E)$, so $S(E)$ is simply $E$ itself. The following steps are performed in Phase III using $S_0 = S(E)$:\\

\noindent $1'$: $f(ab)\approx f(ac)$ (DEDUCE by $f(ab)\approx a$ and $f(bc)\approx b$.)

\noindent $2'$: $f(ac) \approx a$ (SIMPLIFY $1'$ by $f(ab)\approx a$. $1'$ is deleted.)

\noindent $3'$: $f(ad) \approx a$ (COLLAPSE $2'$ by $c\approx d$. $2'$ is deleted.)

\noindent $4'$: $f(bd) \approx b$ (COLLAPSE $f(b,c)\approx b$ by $c\approx d$. $f(b,c)\approx b$ is deleted.)\\

Now, we have $S_\infty^\succ(E) = \{f(ab) \rightarrow a, f(bd)\rightarrow b, c\rightarrow d, f(ad) \rightarrow a\}$. By Lemma~\ref{lem:ccequiv_semigroups} and Corollary~\ref{cor:wordproblem_semigroups}, we can decide whether $f(a,b) \approx_E^A f(a,c)$. Since $f(ab) \rightarrow_{S_\infty^\succ(E)/A} a$ and $f(ac) \rightarrow_{S_\infty^\succ(E)/A} f(ad)\rightarrow_{S_\infty^\succ(E)/A} a$, we see that $f(a,b) \approx f(a,c) \in CC^A(E)$ and $f(a,b) \approx_E^A f(a,c)$.
\end{exa}

The following example shows that $S_\infty(E)$ can be infinite. In this case, one cannot apply Corollary~\ref{cor:wordproblem_semigroups}.

\begin{exa}\label{ex:nonterminating_ac}\normalfont (Adapted from~\cite{Kapur1985}) Let $E=\{f(a,b, a)\approx f(b,a,b)\}$ with $f\in \Sigma_A$. Each term in $E$ is already associatively flattened. Also, each equation in $E$ is already a ground fully flat equation, so Phase I is not needed. Phase II is not needed either for $CC^A(E)$, so $S(E)$ is simply $E$ itself. The following steps are performed in Phase III using $S_0 = S(E)$ with $a\succ b$:\\

\noindent $1'$: $f(abbab)\approx f(babba)$ (DEDUCE by $f(aba)\approx f(bab)$ and itself.)

\noindent $2'$: $f(abbbab) \approx f(babbaa)$ (DEDUCE by $1'$ and $f(aba)\approx f(bab)$.)

$\cdots$
\\
Phase III does not terminate and we have the infinite $S_\infty^\succ(E) = \{f(aba) \rightarrow f(bab)\} \cup \{f(ab^nab)\rightarrow f(babba^{n-1})\,|\,n\geq 2\}$. Using the similar argument by Kapur and Narendran~\cite{Kapur1985}, we see that  Phase III does not terminate. (Here, the word problem for $E$ w.r.t.\;$A$ is decidable~\cite{Kapur1985}.)
\end{exa}

\begin{rem} In Example~\ref{ex:nonterminating_ac}, if we introduced a new constant $c_1$ to stand for $f(a,b)$ with $c_1 \succ b\succ a$, then we have $E' = \{f(a, b) \approx c_1, f(c_1, a) \approx f(b, c_1)\}$ from $E$. In this case, Phase III using $S_0 = E'$ terminates with the finite $S_\infty^\succ(E')=\{f(ab) \to c_1, f(c_1a) \to f(bc_1), f(bc_1b) \to f(c_1c_1), f(c_1c_1b) \to f(ac_1c_1)\}$ (cf. Section 6 in~\cite{Kapur1985}). This is beyond the scope of this paper because we do not introduce a new constant for an (already) ground flat equation in $E$ for $CC^A(E)$.
\end{rem}

\subsection{Congruence closure modulo monoids}
In this subsection, we denote by $\Sigma_M$ the set $\{f, 1\}$, where $f\in \Sigma_A$ is the interpreted symbol for the monoid axioms $M:=A\cup \{f(x,1)\approx x, f(1,x)\approx x\}$ and $1 \in C_0$ is the unit in $M$. The convergent rewrite system for monoids, denoted $R(M)$, on associatively flat terms is simply given as follows: (i) $f(x,1) \rightarrow x$ and (ii) $f(1,x)\rightarrow x$, where $f\in \Sigma_M$. In what follows, we assume
that the multiple associative symbols are allowed (i.e.,\;$|\Sigma_A| \geq 1$), but only one of them is used for the monoid axioms. 
Phases I and II for constructing congruence modulo monoids are slightly different from Phases I and II in Section~\ref{sec:cc}, which do not need to take the inverse axioms into account. In Phase I, the rewrite relation $R(G)/A$ in step 2 
in Phases I in Section~\ref{sec:cc} is simply replaced by the rewrite relation $R(M)/A$. The output of Phase I is denoted by $E'$, where all equations in $E'$ are constant, $D$-flat, or $A$-flat equations. The purpose of Phase II is to add certain ground instantiations of the unit axioms.  Phase II is now described as follows. Here, $f\in \Sigma_M$, and 1 is the unit in $M$.\\

\noindent {\bf Phase II}: Given $C$ and $E'$ obtained from $E$ by Phase I:
\begin{enumerate}
\item Set $S(E):= E' \cup U(C)$ and return $S(E)$, where $U(C):=\{f(c,1)\approx c\,|\,c\in C\} \cup \{f(1,c)\approx c\,|\,c\in C\}$.
\end{enumerate}

The output of Phase II is $S(E)$. Note that no new constant is added to $C$ in Phase II. Using $S(E)$ obtained from Phase II, we may apply the same Phase III as in Section~\ref{sec:cc} using $S_0 = S(E)$. Now, the following results are adapted from Lemma~\ref{lem:cctheory}, Theorem~\ref{thm:ccequiv} and Corollary~\ref{cor:wordproblem}, respectively.

\begin{lem}\label{lem:cctheory_monoids} 
If $s_0, t_0 \in T(\Sigma, C_0)$, then $s_0 \approx t_0 \in CC^M(E)$ iff $s_0\approx_{S_\infty(E) \cup gr(M)\cup gr(A)} t_0$.
\end{lem}
\begin{proof}
By Birkhoff's theorem, $CC^M(E)$ coincides with $\approx_E^M$. We have $s_0\approx_E^M t_0$ iff $s_0\approx_{S(E)}^M t_0$ for all terms $s_0, t_0 \in T(\Sigma, C_0)$ using a simple adaption of Lemma~\ref{lem:conservative}. Also, $s_0\approx_{S(E)}^M t_0$ iff $s_0\approx_{S(E) \cup gr(R(M))\cup gr(A)} t_0$ for all terms $s_0, t_0 \in T(\Sigma, C_0)$. 

Now, it remains to show that  $s_0\approx_{S(E) \cup gr(R(M))\cup gr(A)} t_0$ iff $s_0\approx_{S_\infty(E) \cup gr(R(M))\cup gr(A)} t_0$ for all terms $s_0, t_0 \in T(\Sigma, C_0)$, where $S_0=S(E)$. By Lemma~\ref{lem:equiv}, if $S_i \vdash S_{i+1}$, then both $\xleftrightarrow{*}_{S_i\cup gr(A)}$ and $\xleftrightarrow{*}_{S_{i+1}\cup gr(A)}$ on $T(\Sigma, C_0)$ are the same, and thus the conclusion follows.
\end{proof}

\begin{lem}\label{lem:ccequiv_monoids} If $s_0, t_0 \in T(\Sigma, C_0)$, then $s_0 \approx t_0 \in CC^M(E)$ iff $s_0$ and $t_0$ have the same normal form w.r.t.\;$(S_\infty^\succ(E) \cup gr(R(M))/A$.
\end{lem}
\begin{proof}
Assume that $s_0, t_0 \in T(\Sigma, C_0)$. If $s_0\approx t_0 \in CC^M(E)$, then $s_0 \approx_{S_\infty(E) \cup gr(R(M))\cup gr(A)} t_0$ by Lemma~\ref{lem:cctheory_monoids}. Using a simple adaptation from Theorem~\ref{thm:canonical} without taking the symbol for the inverse axioms into account, we may infer that $S_{\infty}^\succ(E) \cup gr(R(M))$ is convergent modulo $A$, and thus $s_0$ and $t_0$ have the same normal form w.r.t.\;$(S_\infty^\succ(E) \cup gr(R(M))/A$.

Conversely, if $s_0$ and $t_0$ have the same normal form w.r.t.\;$(S_\infty^\succ(E) \cup gr(R(M))/A$, then we have  $s_0 \approx_{S_\infty(E) \cup gr(R(M)) \cup gr(A)} t_0$, and thus $s_0\approx t_0 \in CC^M(E)$  by Lemma~\ref{lem:cctheory_monoids}.
\end{proof}

\begin{cor}\label{cor:wordproblem_monoids} Given a finite set of ground equations $E\subseteq T(\Sigma, C_0) \times T(\Sigma, C_0)$, if $S_\infty(E)$ is finite,  then we can decide for any $s_0, t_0 \in T(\Sigma, C_0)$ whether $s_0 \approx_E^M t_0$ holds or not.
\end{cor}

\begin{exa}[Continued from Example~\ref{ex:ex7}]\label{ex:ex8}\normalfont Consider $E=\{f(a,b)\approx a, f(b,c)\approx b, c\approx d\}$ with $f\in \Sigma_A$ and $a\succ b \succ c \succ d$ in Example~\ref{ex:ex7} again. Again, Phase I is not needed. For Phase II, since $C = C_0 = \{a,b,c,d,1\}$, we have $U(C) = \{f(1,1)\approx 1, f(a,1)\approx a, f(1,a)\approx a, f(b,1)\approx b, f(1,b)\approx b, f(c,1)\approx c, f(1,c)\approx c, f(d,1)\approx d, f(1,d)\approx d\}$. Now, $S(E) = E \cup U(C)$ and the same equations as in Example~\ref{ex:ex7} are generated by Phase III after some contraction steps. Therefore, $S_\infty^\succ(E) = \{f(ab) \rightarrow a, f(bd)\rightarrow b, c\rightarrow d, f(ad) \rightarrow a, f(11)\rightarrow 1, f(a1)\rightarrow a, f(1a)\rightarrow a, f(b1)\rightarrow b, f(1b)\rightarrow b, f(c1)\rightarrow c, f(1c)\rightarrow c, f(d1)\rightarrow d, f(1d)\rightarrow d\}$. By Lemma~\ref{lem:ccequiv_monoids} and Corollary~\ref{cor:wordproblem_monoids}, we can decide whether $f(a,c,1,d) \approx_E^A a$. Since $f(ac1d) \rightarrow_{(S_\infty^\succ(E)\cup gr(R(M)))/A} f(acd)\rightarrow_{(S_\infty^\succ(E)\cup gr(R(M)))/A} f(add) \rightarrow_{(S_\infty^\succ(E)\cup gr(R(M)))/A} f(ad) \rightarrow_{(S_\infty^\succ(E)\cup gr(R(M)))/A} a$, we see that $f(a,c,1,d) \approx a \in CC^M(E)$ and $f(a,c,1,d) \approx_E^M a$.
\end{exa}

\subsection{Congruence closure modulo the multiple sets of group axioms} Section~\ref{sec:cc} was concerned with congruence closure modulo a single set of group axioms. This subsection adapts Section~\ref{sec:cc} for constructing congruence closure of a finite set of ground equations with interpreted symbols for the multiple disjoint sets of group axioms.\footnote{Recall that a single set of group axioms $G$ has the following form in this paper: $G:=A\cup \{f(x,1)\approx x, f(1,x)\approx x, f(x, i(x))\approx 1, f(i(x),x)\approx 1\}$. Alternatively, there are single axioms for groups~\cite{McCune1993}. In this paper, by a set of group axioms, we mean a set of group axioms having the form of $G$ shown above.}

First, we consider the union of two sets of group axioms $G_1$ and $G_2$ with\;$\Sigma_{G_1} \cap \Sigma_{G_2} = \emptyset$, denoted by $G_1\uplus G_2$, where $\Sigma_{G_1}=\{f_{G_1},i_{G_1}, 1_{G_1}\}$ and $\Sigma_{G_2}=\{f_{G_2},i_{G_2}, 1_{G_2}\}$. This means that both $G_1$ and $G_2$ are the sets of group axioms but they do not share any function symbols. Let $R(G_1)$ and $R(G_2)$ be the convergent rewrite systems for $G_1$ and $G_2$ on associatively flat terms, respectively (see Section~\ref{sec:preliminaries}). Now, the union of two rewrite systems $R_1$ and $R_2$ for $G_1\uplus G_2$, denoted by $R(G_1) \uplus R(G_2)$, on associatively flat terms has the following rewrite rules:

\begin{longtable}{lll}
$i_{G_1}(1_{G_1})\rightarrow 1_{G_1}$  & $f_{G_1}(x,1_{G_1}) \rightarrow x$ & $f_{G_1}(1_{G_1},x)\rightarrow x$ \;\;\;\;\; $i_{G_1}(i_{G_1}(x))\rightarrow x$\\ 
$f_{G_1}(x, i_{G_1}(x))\rightarrow 1_{G_1}$ & $f_{G_1}(i_{G_1}(x),x)\rightarrow 1_{G_1}$ & $i_{G_1}(f_{G_1}(x,y))\rightarrow f_{G_1}(i_{G_1}(y), i_{G_1}(x))$\\\\
$i_{G_2}(1_{G_2})\rightarrow 1_{G_2}$  & $f_{G_2}(x,1_{G_2}) \rightarrow x$ & $f_{G_2}(1_{G_2},x)\rightarrow x$ \;\;\;\;\; $i_{G_2}(i_{G_2}(x))\rightarrow x$\\ 
$f_{G_2}(x, i_{G_2}(x))\rightarrow 1_{G_2}$ & $f_{G_2}(i_{G_2}(x),x)\rightarrow 1_{G_2}$ & $i_{G_2}(f_{G_2}(x,y))\rightarrow f_{G_2}(i_{G_2}(y), i_{G_2}(x))$
\end{longtable}

It is known that termination is not a \emph{modular property}~\cite{Ohlebusch1994} of rewrite systems, while confluence is a modular property of rewrite systems~\cite{Ohlebusch1994,Toyama1987}. It is also known that the disjoint union of two disjoint rewrite systems $R_1$ and $R_2$ is terminating if neither $R_1$ nor $R_2$ contains \emph{duplicating rules}~\cite{Ohlebusch1994,Rusinowitch1987}. (A rewrite rule $l\rightarrow r$ is \emph{duplicating}~\cite{Ohlebusch1994} if there exists some variable such that it has more occurrences in $r$ than $l$.) In the above example, the union of $R(G_1)$ and $R(G_2)$ is confluent and terminating on associatively flat terms because $\Sigma_{G_1} \cap \Sigma_{G_2} = \emptyset$ and neither $R(G_1)$ nor $R(G_2)$ contains duplicating rules. In the remainder of this subsection, we denote by $\biguplus_{i=1}^nG_i:=G_1\uplus\cdots \uplus G_n$  the union of the sets of group axioms $G_1,\ldots,G_n$ such that $\Sigma_{G_i} \cap \Sigma_{G_j} = \emptyset$ for every $i\neq j$. We denote by $\biguplus_{i=1}^n R(G_i):= R(G_1)\uplus \cdots \uplus R(G_n)$ the union of $R(G_1),\ldots, R(G_n)$, where $R(G_k)$, $1\leq k \leq n$, is the convergent rewrite system for $G_k$ on associatively flat terms. Now, the next lemma follows from the above observation.

\begin{lem}\label{lem:modular} $\biguplus_{i=1}^n R(G_i)$ is convergent modulo $A$.
\end{lem} 

Phases I and II for constructing congruence closure modulo $\biguplus_{i=1}^n G_i$ are adapted from Phases I and II in Section~\ref{sec:cc} by taking different interpreted symbols for the multiple disjoint sets of group axioms into account. The rewrite relation $R(G)/A$ in step 2 
in Phases I in Section~\ref{sec:cc} is simply replaced by the rewrite relation $\biguplus_{i=1}^n R(G_i)/A$. The output of Phase I is denoted by $E'$, where all equations in $E'$ are constant, $D$-flat, or $A$-flat equations.

Next, Phase II is described as follows for $\biguplus_{i=1}^nG_i$ and their interpreted symbols $f_{G_l}, i_{G_l}, 1_{G_l}$ for all $1 \leq l \leq n$.\\

\noindent {\bf Phase II}: Given $C$ and $E'$ obtained from $E$ by Phase I, Copy $C$ to $C'$ and $E'$ to $E^{''}$.

\noindent Then, for each set of group axioms $G_l$, $1\leq l \leq n$, repeat the following procedure:
\begin{itemize}
\item For each constant $c_k\in C'$ and $c_k \neq 1_{G_l}$, repeat the following step:

\noindent If neither $i_{G_l}(c_k)\approx c_i$ nor $i_{G_l}(c_j)\approx c_k$ appears in $E^{''}$ for some $c_i,c_j\in C'$, then $ E':=E'\cup \{i_{G_l}(c_k)\approx c_m\}$ and $C:=C\cup \{c_m\}$ for a new constant $c_m$ taken from $W$.
\end{itemize}

\noindent Next, for each set of group axioms $G_l$, $1\leq l \leq n$, repeat 
 the following steps:
\begin{itemize}
\item Set $I_{G_l}(E'):=\{i_{G_l}(1_{G_l})\approx 1_{G_l}\} \cup \{i_{G_l}(c_n)\approx c_m\,|\,i_{G_l}(c_m)\approx c_n\in E'\} \cup \{f_{G_l}(c_m, c_n)\approx 1_{G_l}\,|\,i_{G_l}(c_m)\approx c_n\in E'\}\cup \{f_{G_l}(c_n, c_m)\approx 1_{G_l}\,|\,i_{G_l}(c_m)\approx c_n\in E'\}$.
\item Set $U_{G_l}(C):=\{f_{G_l}(c,1_{G_l})\approx c\,|\,c\in C\} \cup \{f_{G_l}(1_{G_l},c)\approx c\,|\,c\in C\}$.\\
\end{itemize}

\noindent Finally, set $S(E):= E' \cup I_{G_1}(E')\cup\cdots\cup I_{G_n}(E')\cup U_{G_1}(C)\cup \cdots \cup U_{G_n}(C)$ and return $S(E)$. Here, $S(E)$ is the output of Phase II. Now, using $S(E)$ obtained from Phase II, one may apply the same Phase III as in Section~\ref{sec:cc} using $S_0 = S(E)$. The following lemma is a direct extension of Lemma~\ref{lem:conservative}.

\begin{lem}\label{lem:cons_mul} Viewed as a set of equations, $S(E)$ w.r.t.\;$\widehat{G}:=\biguplus_{i=1}^n G_i$ is a conservative extension of $E$ w.r.t.\;$\widehat{G}$, i.e., $s_0\approx_E^{\widehat{G}} t_0$ iff $s_0\approx_{S(E)}^{\widehat{G}} t_0$ for all terms $s_0, t_0 \in T(\Sigma, C_0)$.
\end{lem}

Similarly, the following results are adapted from Lemma~\ref{lem:cctheory}, Theorem~\ref{thm:ccequiv} and Corollary~\ref{cor:wordproblem}, respectively.

\begin{lem}\label{lem:cctheory_mul} Let $\widehat{G}:=\biguplus_{i=1}^n G_i$ and $R(\widehat{G}):=\biguplus_{i=1}^n R(G_i)$.
If $s_0, t_0 \in T(\Sigma, C_0)$, then $s_0 \approx t_0 \in CC^{\widehat{G}}(E)$ iff $s_0\approx_{S_\infty(E) \cup gr(R(\widehat{G}))\cup gr(A)} t_0$.
\end{lem}
\begin{proof}
First, $CC^{\widehat{G}}(E)$ is the same as $\approx_E^{\widehat{G}}$ by Birkhoff's theorem. Also, by Lemma~\ref{lem:cons_mul}, $s_0\approx_E^{\widehat{G}} t_0$ iff $s_0\approx_{S(E)}^{\widehat{G}}t_0$ for all terms $s_0, t_0 \in T(\Sigma, C_0)$. Then, $s_0\approx_{S(E)}^{\widehat{G}}t_0$ iff $s_0\approx_{S(E) \cup gr(R(\widehat{G}))\cup gr(A)} t_0$ for all terms $s_0, t_0 \in T(\Sigma, C_0)$. 

Let $S_0=S(E)$. Applying Lemma~\ref{lem:equiv} yields that if $S_i \vdash S_{i+1}$, then $\xleftrightarrow{*}_{S_i\cup gr(A)}$ and $\xleftrightarrow{*}_{S_{i+1}\cup gr(A)}$ on $T(\Sigma, C_0)$ are the same, and thus $\xleftrightarrow{*}_{S(E)\cup gr(A)}$ and $\xleftrightarrow{*}_{S_\infty(E)\cup gr(A)}$ coincide. Thus, $s_0\approx_{S(E) \cup gr(R(\widehat{G}))\cup gr(A)} t_0$ iff $s_0\approx_{S_\infty(E) \cup gr(R(\widehat{G}))\cup gr(A)} t_0$ for all terms $s_0, t_0 \in T(\Sigma, C_0)$. 
\end{proof}

\begin{lem}\label{lem:ccequiv_mul} Let $\widehat{G}:=\biguplus_{i=1}^n G_i$ and $R(\widehat{G}):=\biguplus_{i=1}^n R(G_i)$. If $s_0, t_0 \in T(\Sigma, C_0)$, then $s_0 \approx t_0 \in CC^{\widehat{G}}(E)$ iff $s_0$ and $t_0$ have the same normal form w.r.t.\;$(S_\infty^\succ(E) \cup gr(R({\widehat{G}})))/A$.
\end{lem}
\begin{proof}
Assume that $s_0, t_0 \in T(\Sigma, C_0)$. If $s_0 \approx t_0 \in CC^{\widehat{G}}(E)$, then $s_0\approx_{S_\infty(E) \cup gr(R(\widehat{G}))\cup gr(A)} t_0$ for all terms $s_0, t_0 \in T(\Sigma, C_0)$ by Lemma~\ref{lem:cctheory_mul}. By  Theorem~\ref{thm:canonical}, we may infer that each $S_{\infty}^\succ(E) \cup gr(R(G_l))$, $1\leq l \leq n$, is convergent modulo $A$. Now by Lemma~\ref{lem:modular} and a simple adaptation of Theorem~\ref{thm:canonical}, we may also infer that $S_{\infty}^\succ(E) \cup gr(R(\widehat{G}))$ is convergent modulo $A$, and thus $s_0$ and $t_0$ have the same normal form w.r.t.\;$(S_\infty^\succ(E) \cup gr(R(\widehat{G}))/A$.

Conversely, if $s_0$ and $t_0$ have the same normal form w.r.t.\;$(S_\infty^\succ(E) \cup gr(R(\widehat{G}))/A$, then we have  $s_0 \approx_{S_\infty(E) \cup gr(R(\widehat{G})) \cup gr(A)} t_0$, and thus $s_0\approx t_0 \in CC^{\widehat{G}}(E)$  by Lemma~\ref{lem:cctheory_mul}.
\end{proof}

\begin{cor}\label{cor:wordproblem_mul} Let $\widehat{G}:=\biguplus_{i=1}^n G_i$.
Given a finite set of ground equations $E\subseteq T(\Sigma, C_0) \times T(\Sigma, C_0)$, if $S_\infty(E)$ is finite,  then we can decide for any $s_0, t_0 \in T(\Sigma, C_0)$ whether $s_0 \approx_E^{\widehat{G}}t_0$ holds or not.
\end{cor}

\begin{exa}\label{ex:ex9}\normalfont Let $G_1$ and $G_2$ be the sets of group axioms with $\Sigma_{G_1} = \{f, i_f, 1_f\}$ and $\Sigma_{G_2} = \{g, i_g, 1_g\}$. Consider $E=\{f(a,b)\approx a, f(b,a)\approx b, g(a,b) \approx g(b,a), h(a)\approx b\}$ with $a\succ b\succ 1_f \succ 1_g$ and $h \in \Sigma$. Each equation in $E$ is already a ground fully flat equation, so Phase I is not needed. (Note that equation $g(a,b) \approx g(b,a)$ is an $A$-flat equation.) After the first step of Phase II for $G_1$ and $G_2$, we have:

\indent $E'= \{f(a,b)\approx a, f(b,a)\approx b, g(a,b) \approx g(b,a), h(a)\approx b\} \cup \{i_f(a) \approx c_1, i_f(b) \approx c_2, i_f(1_g) \approx c_3, i_g(a) \approx c_4, i_g(b) \approx c_5, i_g(1_f) \approx c_6\}$ and $C= C_0\cup C_1$, where $C_0 =\{a,b,1_f, 1_g\}$ and $C_1=\{c_1,c_2,c_3,c_4, c_5, c_6\}$.

\indent After the remaining steps of Phase II for $G_1$ and $G_2$, we have $I_{G_1}(E') =\{i_f(1_f)\approx 1_f, i_f(c_1)\approx a, i_f(c_2)\approx b, i_f(c_3)\approx 1_g, f(c_1,a)\approx 1_f, f(a,c_1)\approx 1_f,f(c_2, b)\approx 1_f, f(b, c_2)\approx 1_f, f(c_3, 1_g) \approx 1_f, f(1_g, c_3)\approx 1_f\}$ and $I_{G_2}(E')=\{i_g(1_g)\approx 1_g, i_g(c_4)\approx a, i_g(c_5)\approx b, i_g(c_6)\approx 1_f, g(c_4, a)\approx 1_g, g(a, c_4)\approx 1_g, g(c_5,b)\approx 1_g,  g(b, c_5)\approx 1_g, g(c_6, 1_f) \approx 1_g, g(1_f, c_6)\approx 1_g\}$.

Also, we have $U_{G_1}(C) = \{f(c, 1_f)\approx c\,|\,c\in C\}\cup \{f(1_f, c)\approx c\,|\,c\in C\}$ and $U_{G_2}(C) = \{g(c, 1_g)\approx c\,|\,c\in C\}\cup \{g(1_g, c)\approx c\,|\,c\in C\}$.

\indent After Phase II, $S(E):= E' \cup I_{G_1}(E')\cup I_{G_2}(E')\cup U_{G_1}(C)\cup U_{G_2}(C)$. Now, the following steps are performed in Phase III using $S_0=S(E)$.\\

\noindent $1$: $f(ac_2)\approx f(a1_f)$ (DEDUCE by $f(ab)\approx a$ and $f(bc_2)\approx 1_f$.)

\noindent $2$: $f(ac_2) \approx a$ (SIMPLIFY $1$ by $f(a1_f)\approx a$. $1$ is deleted.)

\noindent $3$: $f(bc_2)\approx f(ba)$ (DEDUCE $2$ by $f(ba)\approx b$.)

\noindent $4$: $f(ba) \approx 1_f$ (SIMPLIFY $3$ by $f(bc_2)\approx 1_f$. $3$ is deleted.)

\noindent $5$: $b \approx 1_f$ (SIMPLIFY $4$ by $f(ba)\approx b$. $4$ is deleted.)

\noindent $6$: $f(1_fc_2) \approx 1_f$ (COLLAPSE $f(bc_2)\approx 1_f$ by $b\approx 1_f$. $f(bc_2)\approx 1_f$ is deleted.)

\noindent $7$: $c_2 \approx 1_f$ (SIMPLIFY $6$ by $f(1_fc_2)\approx c_2$. $6$ is deleted.)

\noindent $8$: $f(1_fa) \approx b$ (COLLAPSE $f(ba) \approx b$ by $b\approx 1_f$. $f(ba) \approx b$ is deleted.)

\noindent $9$: $f(1_fa) \approx 1_f$ (COMPOSE $8$ by $b\approx 1_f$. $8$ is deleted.)

\noindent $10$: $a \approx 1_f$ (SIMPLIFY $9$ by $f(1_fa)\approx a$. $9$ is deleted.)

\noindent $11$: $f(1_fc_1) \approx 1_f$ (COLLAPSE $f(ac_1)\approx 1_f$ by $a\approx 1_f$. $f(ac_1)\approx 1_f$ is deleted.)

\noindent $12$: $c_1 \approx 1_f$ (SIMPLIFY $11$ by $f(1_fc_1)\approx c_1$. $11$ is deleted.)

\noindent $13$: $i_g(1_f)\approx c_4$ (COLLAPSE $i_g(a)\approx c_4$ by $a\approx 1_f$. $i_g(a)\approx c_4$ is deleted.)

\noindent $14$: $i_g(1_f)\approx c_5$ (COLLAPSE $i_g(b)\approx c_5$ by $b\approx 1_f$. $i_g(b)\approx c_5$ is deleted.)

\noindent $15$: $c_4\approx c_5$ (COLLAPSE $13$ by $14$. $13$ is deleted.)

\noindent $16$: $c_5\approx c_6$ (COLLAPSE $14$ by $i_g(1_f) \approx c_6$. $14$ is deleted.)
\\
$\cdots$\\

After several steps using the contraction rules, we have $S_\infty^\succ(E) = \{a\rightarrow 1_f, b\rightarrow 1_f, c_1\rightarrow 1_f, c_2 \rightarrow 1_f,  h(1_f)\rightarrow 1_f, i_f(1_f) \rightarrow 1_f, c_4\rightarrow c_6, c_5\rightarrow c_6, i_f(1_g)\rightarrow c_3, i_f(c_3) \rightarrow 1_g, f(c_3, 1_g) \rightarrow 1_f, f(1_g, c_3)\rightarrow 1_f, i_g(1_f) \rightarrow c_6, i_g(c_6)\rightarrow 1_f, i_g(1_g)\rightarrow 1_g, g(c_6, 1_f)\rightarrow 1_g, g(1_f, c_6)\rightarrow 1_g\} \cup \bar{U}^\succ(C)$, where $\bar{U}^\succ(C)=\{f(1_f,1_f) \rightarrow 1_f, f(1_f, c_3)\rightarrow c_3, f(c_3, 1_f)\rightarrow c_3, f(1_f, c_6)\rightarrow c_6, f(c_6, 1_f)\rightarrow c_6,  f(1_f, 1_g)\rightarrow 1_g, f(1_g, 1_f)\rightarrow 1_g\} \cup \{g(1_g,1_g) \rightarrow 1_g, g(1_g, c_3)\rightarrow c_3, g(c_3, 1_g)\\\rightarrow c_3, g(1_g, c_6)\rightarrow c_6, g(c_6, 1_g)\rightarrow c_6,  g(1_g, 1_f)\rightarrow 1_f, g(1_f, 1_g)\rightarrow 1_f\}$. By Lemma~\ref{lem:ccequiv_mul} and Corollary~\ref{cor:wordproblem_mul}, we can decide whether $g(f(a,a),h(b))\approx_E^{G_1 \uplus G_2} g(b,f(a,1_f))$ holds or not. Let $MG := S_\infty^\succ(E) \cup gr(R(G_1) \uplus R(G_2))$. Then, we see that $g(f(a,a),h(b))\xrightarrow{*}_{MG/A} g(f(1_f,1_f),h(1_f))\xrightarrow{*}_{MG/A}g(1_f,1_f)$. Also, $g(b,f(a,1_f))\rightarrow_{MG/A} g(b,a)\xrightarrow{*}_{MG/A} g(1_f, 1_f)$. Now, we conclude that $g(f(a,a),h(b))\approx_E^{G_1 \uplus G_2} g(b,f(a,1_f))$ holds and $g(f(a,a),h(b))\approx g(b,f(a,1_f))\in CC^{G_1 \uplus G_2}(E)$ by Lemma~\ref{lem:ccequiv_mul}.
\end{exa}

\section{Conclusion}
This paper has presented a new framework for computing congruence closure of a finite set of ground equations $E$ over uninterpreted symbols and interpreted symbols for a set of group axioms $G$, which extends a rewrite-based congruence closure procedure in~\cite{Kapur1997} by taking $G$ into account. In the proposed framework, ground equations in $E$ are flattened into ground flat equations and certain ground flat equations entailed by $G$ are added for a completion procedure. The proposed completion procedure is a ground completion procedure using strings, which adapts a completion procedure for string rewriting systems presenting groups~\cite{Holt2005, Sims1994}. The procedure yields a ground convergent rewrite system for congruence closure modulo $G$ for $E$. It is simple and generic in the sense that it can also be used for constructing congruence closure w.r.t.\;the semigroup, monoid, and the multiple disjoint sets of group axioms for $E$ by a slight change of Phase II, but without changing the completion procedure itself. Furthermore, it neither uses extension rules nor complex orderings.

In~\cite{Rubio1995},~\cite{Bachmair1991}, and~\cite{Peterson1981}, the authors pointed out that it is natural to view the associativity axiom (with a binary function symbol) as a ``structural axiom'' rather than viewing it as a ``simplifier''. For example, from the given following simple rules $f(f(x,y),z) \rightarrow f(x,f(y,z))$, $f(a,b)\rightarrow b$, $f(a,f(x,b))\rightarrow f(x,b)$, the standard completion may generate infinite rules $f(a,f(x,f(y,b))) \rightarrow f(x,f(y,b)), f(a,f(x,f(y,f(z,b))))\rightarrow f(x,f(y,f(z,b))), \cdots$~\cite{Bachmair1991,Rubio1995}. Accordingly, this paper is concerned with a completion procedure for the rewrite relation $\rightarrow_{R/A}$ on associatively flat ground terms instead of the (plain) rewrite relation $\rightarrow_{R\cup A}$. Also, the proposed completion procedure does not use infinitary $A$-unification explicitly and simply use string matching for ground flat terms during the proposed (ground) completion procedure; note that, for example, $f(a,x)$ and $f(x,a)$ have an infinite set of (independent) $A$-unifiers $a$, $f(a,a)$, $f(a,a,a)$, etc. (See~\cite{Baader2001} and~\cite{Rubio1995}.)

Meanwhile, completion modulo $A$ was considered in~\cite{Rubio1995}, which is not well suited for constructing congruence closure modulo $A$. Besides, the approach uses the complex $A$-compatible reduction ordering and $A$-unification.

Completion modulo a set of axioms $Ax$ for left-linear rules was considered in~\cite{Huet1980}. This approach is also not suited for constructing congruence closure modulo $G$ because $R(G)$ has a non-left linear rule, for example, $f(x, i(x)) \rightarrow 1$.

In~\cite{Kandri85}, the authors discussed an approach based on computing the Gr\"obner basis of polynomial ideals for solving word problems for finitely presented commutative algebraic structures including commutative rings with unity. This approach is also not suited for constructing congruence closure modulo $G$ because the commutativity is not assumed in $G$.

The new results and the main contributions of this paper are the construction of a rewriting-based congruence closure of a finite set of ground equations modulo the semigroup, monoid, group, and the multiple disjoint sets of group axioms, respectively. No nonground equation is added for the above constructions during the proposed completion procedure, and they are all represented by convergent ground rewrite systems modulo associativity. In particular, if the proposed completion procedure terminates for a finite set of ground equations w.r.t.\;the semigroup, monoid, group, or the multiple disjoint sets of group axioms, then it yields a decision procedure for the word problem for a finite set of ground equations w.r.t.\;the corresponding axioms. Recall that the word problem for finitely presented semigroups, monoids, and groups are undecidable in general~\cite{Book1993,Holt2005}.

The key insights of this paper are as follows. First, the arguments of an associatively flat ground term headed by each associative symbol $f$ are represented by the corresponding string, which is possible by flattening nonflat ground terms occurring in the arguments of $f$ by introducing new constants. Also, certain ground flat equations entailed by the group axioms (or the monoid axioms) are added for the proposed \emph{ground} completion procedure so that one needs neither an $A$-compatible reduction ordering nor $A$-unification. This also allows the well-known completion procedure for string rewriting systems and its results (e.g.\;a monoid presentation of a finite group) to be adapted for the proposed framework of constructing congruence closure of ground equations modulo the semigroup, monoid, group, and the multiple disjoint sets of group axioms, respectively. It is left as a future work to translate the (known) sufficient termination criteria for the completion of string rewriting systems (or \emph{Thue systems})~\cite{KapurN85,Book1993,Polina1998} into the termination criteria of constructing congruence closure of ground equations modulo the semigroup, monoid, and the group axioms, respectively.

Since groups are one of the fundamental objects in mathematics, physics, and computer science, developing applications (e.g.\;integration of SMT solvers with the proposed approaches) using the proposed approaches remains as future research opportunities. 

Finally, some of the potential extension of the results discussed in this paper can be the construction of a rewriting-based congruence closure of a finite set of ground equations modulo the following: (i) 
 \emph{involutive}~\cite{Easdown1993} semigroups, monoids, and groups, (ii) \emph{idempotent}~\cite{Siekmann1982} semigroups, monoids, and groups, and (iii) \emph{cancellative}~\cite{Narendran1989} semigroups, monoids, and groups, and their combinations. Construction of a rewriting-based congruence closure of a finite set of ground equations modulo two sets of group axioms with \emph{homomorphisms}~\cite{Book1993} can be another potential extension of the results discussed in this paper.
\section*{Acknowledgment}
The author would like to thank Christina Kirk, Aart Middeldorp, Fabian Mitterwallner, Johannes Niederhauser, Teppei Saito, and Jonas Sch{\"{o}}pf (in alphabetical order) for their valuable comments and feedback of this work during the regular TRS reading group meeting of the Computational Logic group at the University of Innsbruck, Austria. 

The author also would like to thank the reviewers for their insightful suggestions and comments, which helped to improve this paper.
\bibliographystyle{alphaurl}
\bibliography{main}

\begin{thebibliography}{CRSS94}

\bibitem[Bac91]{Bachmair1991}
Leo Bachmair.
\newblock {\em Canonical {E}quational {P}roofs}.
\newblock Birkh{\"a}user, Boston, 1991.
\newblock \href {https://doi.org/10.1007/978-1-4684-7118-2} {\path{doi:10.1007/978-1-4684-7118-2}}.

\bibitem[BK20]{Baader2020}
Franz Baader and Deepak Kapur.
\newblock Deciding the {W}ord {P}roblem for {G}round {I}dentities with {C}ommutative and {E}xtensional {S}ymbols.
\newblock In N.~Peltier and V.~Sofronie-Stokkermans, editors, {\em Proc. of the 10th Int. Joint Conf. on Automated Reasoning (IJCAR 2020)}, volume 12166 of {\em Lecture Notes in Computer Science}, pages 163--180. Springer, 2020.
\newblock \href {https://doi.org/10.1007/978-3-030-51074-9_10} {\path{doi:10.1007/978-3-030-51074-9_10}}.

\bibitem[BK22]{Baader2022}
Franz Baader and Deepak Kapur.
\newblock Deciding the {W}ord {P}roblem for {G}round and {S}trongly {S}hallow {I}dentities w.r.t. {E}xtensional {S}ymbols.
\newblock {\em J. Autom. Reason.}, 66(3):301--329, 2022.
\newblock \href {https://doi.org/10.1007/S10817-022-09624-4} {\path{doi:10.1007/S10817-022-09624-4}}.

\bibitem[BN98]{Baader1998}
Franz Baader and Tobias Nipkow.
\newblock {\em Term {R}ewriting and {A}ll {T}hat}.
\newblock Cambridge University Press, Cambridge, UK, 1998.
\newblock \href {https://doi.org/10.1017/CBO9781139172752} {\path{doi:10.1017/CBO9781139172752}}.

\bibitem[BO93]{Book1993}
Ronald~V. Book and Friedrich Otto.
\newblock {\em String-{R}ewriting {S}ystems}.
\newblock Springer, New York, NY, 1993.
\newblock \href {https://doi.org/10.1007/978-1-4613-9771-7} {\path{doi:10.1007/978-1-4613-9771-7}}.

\bibitem[BS01]{Baader2001}
Franz Baader and Wayne Snyder.
\newblock {U}nification {T}heory.
\newblock In {\em Handbook of Automated Reasoning}, Volume I, chapter~8, pages 445 -- 532. Elsevier, Amsterdam, 2001.
\newblock \href {https://doi.org/10.1016/B978-044450813-3/50010-2} {\path{doi:10.1016/B978-044450813-3/50010-2}}.

\bibitem[BT18]{Barrett2018}
Clark Barrett and Cesare Tinelli.
\newblock {\em Satisfiability {M}odulo {T}heories}, pages 305--343.
\newblock Springer International Publishing, Cham, 2018.
\newblock \href {https://doi.org/10.1007/978-3-319-10575-8_11} {\path{doi:10.1007/978-3-319-10575-8_11}}.

\bibitem[BTV03]{Tiwari2003}
Leo Bachmair, Ashish Tiwari, and Laurent Vigneron.
\newblock Abstract {C}ongruence {C}losure.
\newblock {\em Journal of Automated Reasoning}, 31(2):129--168, 2003.
\newblock \href {https://doi.org/10.1023/B:JARS.0000009518.26415.49} {\path{doi:10.1023/B:JARS.0000009518.26415.49}}.

\bibitem[Che84]{Chenadec1984}
Philippe~Le Chenadec.
\newblock Canonical {F}orms in {F}initely {P}resented {A}lgebras.
\newblock In {\em Proceedings of the 7th International Conference on Automated Deduction}, page 142–165, Berlin, Heidelberg, 1984. Springer-Verlag.
\newblock \href {https://doi.org/10.1007/978-0-387-34768-4_9} {\path{doi:10.1007/978-0-387-34768-4_9}}.

\bibitem[CRSS94]{Cyrluk1995}
David Cyrluk, Sreeranga Rajan, Natarajan Shankar, and Mandayam~K. Srivas.
\newblock Effective theorem proving for hardware verification.
\newblock In R.~Kumar and T.~Kropf, editors, {\em Theorem Provers in Circuit Design (TPCD 94)}, LNCS 901, pages 203--222. Springer, 1994.
\newblock \href {https://doi.org/10.1007/3-540-59047-1_50} {\path{doi:10.1007/3-540-59047-1_50}}.

\bibitem[DMB11]{Moura2011}
Leonardo De~Moura and Nikolaj Bj{\o}rner.
\newblock Satisfiability modulo theories: introduction and applications.
\newblock {\em Communications of the ACM}, 54(9):69--77, 2011.
\newblock \href {https://doi.org/10.1145/1995376.1995394} {\path{doi:10.1145/1995376.1995394}}.

\bibitem[DP01]{Dershowitz2001}
Nachum Dershowitz and David~A. Plaisted.
\newblock Rewriting.
\newblock In {\em Handbook of Automated Reasoning}, Volume I, chapter~9, pages 535 -- 610. Elsevier, Amsterdam, 2001.
\newblock \href {https://doi.org/10.1016/B978-044450813-3/50011-4} {\path{doi:10.1016/B978-044450813-3/50011-4}}.

\bibitem[DST80]{Downey1980}
Peter~J. Downey, Ravi Sethi, and Robert~Endre Tarjan.
\newblock Variations on the {C}ommon {S}ubexpression {P}roblem.
\newblock {\em J. ACM}, 27(4):758–771, 1980.
\newblock \href {https://doi.org/10.1145/322217.322228} {\path{doi:10.1145/322217.322228}}.

\bibitem[EM93]{Easdown1993}
David Easdown and WD~Munn.
\newblock On semigroups with involution.
\newblock {\em Bulletin of the Australian Mathematical Society}, 48(1):93--100, 1993.
\newblock \href {https://doi.org/10.1017/S0004972700015495} {\path{doi:10.1017/S0004972700015495}}.

\bibitem[HEO05]{Holt2005}
D.~F. Holt, B.~Eick, and E.~A. O'Brien.
\newblock {\em Handbook of computational group theory}.
\newblock CRC Press, Boca Raton, FL, 2005.
\newblock \href {https://doi.org/10.1201/9781420035216} {\path{doi:10.1201/9781420035216}}.

\bibitem[Hue80]{Huet1980}
G{\'{e}}rard~P. Huet.
\newblock {Confluent Reductions: Abstract Properties and Applications to Term Rewriting Systems}.
\newblock {\em J. {ACM}}, 27(4):797--821, 1980.
\newblock \href {https://doi.org/10.1145/322217.322230} {\path{doi:10.1145/322217.322230}}.

\bibitem[Hun80]{Hungerford1980}
Thomas~W. Hungerford.
\newblock {\em Algebra}.
\newblock Springer, New York, NY, 1980.
\newblock \href {https://doi.org/10.1007/978-1-4612-6101-8} {\path{doi:10.1007/978-1-4612-6101-8}}.

\bibitem[Kap97]{Kapur1997}
Deepak Kapur.
\newblock Shostak's {C}ongruence {C}losure as {C}ompletion.
\newblock In H.~Comon, editor, {\em Proc. of the 8th Int. Conf. on Rewriting Techniques and Applications (RTA 1997)}, volume 1232 of {\em Lecture Notes in Computer Science}, pages 23--37. Springer, 1997.
\newblock \href {https://doi.org/10.1007/3-540-62950-5_59} {\path{doi:10.1007/3-540-62950-5_59}}.

\bibitem[Kap21]{Kapur2021}
Deepak Kapur.
\newblock A {M}odular {A}ssociative {C}ommutative ({AC}) {C}ongruence {C}losure {A}lgorithm.
\newblock In N.~Kobayashi, editor, {\em Proc. of 6th International Conference on Formal Structures for Computation and Deduction, {FSCD} 2021}, volume 195, pages 15:1--15:21. LIPIcs, 2021.
\newblock \href {https://doi.org/10.4230/LIPICS.FSCD.2021.15} {\path{doi:10.4230/LIPICS.FSCD.2021.15}}.

\bibitem[Kap23]{Kapur2023}
Deepak Kapur.
\newblock Modularity and {C}ombination of {A}ssociative {C}ommutative {C}ongruence {C}losure {A}lgorithms enriched with {S}emantic {P}roperties.
\newblock {\em Log. Methods Comput. Sci.}, 19(1), 2023.
\newblock \href {https://doi.org/10.46298/LMCS-19(1:19)2023} {\path{doi:10.46298/LMCS-19(1:19)2023}}.

\bibitem[KKN85]{Kandri85}
Abdelilah Kandri{-}Rody, Deepak Kapur, and Paliath Narendran.
\newblock {An Ideal-Theoretic Approach to Word Problems and Unification Problems over Finitely Presented Commutative Algebras}.
\newblock In Jean{-}Pierre Jouannaud, editor, {\em Rewriting Techniques and Applications, First International Conference, RTA-85, Dijon, France, May 20-22, 1985, Proceedings}, volume 202 of {\em Lecture Notes in Computer Science}, pages 345--364. Springer, 1985.
\newblock \href {https://doi.org/10.1007/3-540-15976-2_17} {\path{doi:10.1007/3-540-15976-2_17}}.

\bibitem[KL21]{Kim2021}
Dohan Kim and Christopher Lynch.
\newblock Congruence {C}losure {M}odulo {P}ermutation {E}quations.
\newblock In Temur Kutsia, editor, {\em Proc. of the 9th International Symposium on Symbolic Computation in Software Science, {SCSS} 2021}, volume 342 of {\em {EPTCS}}, pages 86--98, 2021.
\newblock \href {https://doi.org/10.4204/EPTCS.342.8} {\path{doi:10.4204/EPTCS.342.8}}.

\bibitem[KN85a]{Kapur1985}
Deepak Kapur and Paliath Narendran.
\newblock {A Finite Thue System with Decidable Word Problem and without Equivalent Finite Canonical System}.
\newblock {\em Theor. Comput. Sci.}, 35:337--344, 1985.
\newblock \href {https://doi.org/10.1016/0304-3975(85)90023-4} {\path{doi:10.1016/0304-3975(85)90023-4}}.

\bibitem[KN85b]{KapurN85}
Deepak Kapur and Paliath Narendran.
\newblock {The Knuth-Bendix Completion Procedure and Thue Systems}.
\newblock {\em {SIAM} J. Comput.}, 14(4):1052--1072, 1985.
\newblock \href {https://doi.org/10.1137/0214073} {\path{doi:10.1137/0214073}}.

\bibitem[Koz77]{Kozen1977}
Dexter Kozen.
\newblock Complexity of {F}initely {P}resented {A}lgebras.
\newblock In John~E. Hopcroft, Emily~P. Friedman, and Michael~A. Harrison, editors, {\em Proc. of the 9th Annual {ACM} Symposium on Theory of Computing, May 4-6, 1977, Boulder, Colorado, {USA}}, pages 164--177. {ACM}, 1977.
\newblock \href {https://doi.org/10.1145/800105.803406} {\path{doi:10.1145/800105.803406}}.

\bibitem[Mar96]{Marche1996}
Claude March{\'e}.
\newblock {Normalized Rewriting: an Alternative to Rewriting modulo a Set of Equations}.
\newblock {\em Journal of Symbolic Computation}, 21(3):253--288, 1996.
\newblock \href {https://doi.org/10.1006/JSCO.1996.0011} {\path{doi:10.1006/JSCO.1996.0011}}.

\bibitem[McC93]{McCune1993}
William McCune.
\newblock Single {A}xioms for {G}roups and {A}belian {G}roups with {V}arious {O}perations.
\newblock {\em Journal of Automated Reasoning}, 10(1):1--13, 1993.
\newblock \href {https://doi.org/10.1007/BF00881862} {\path{doi:10.1007/BF00881862}}.

\bibitem[NO80]{Nelson1980}
Greg Nelson and Derek~C. Oppen.
\newblock {Fast Decision Procedures Based on Congruence Closure}.
\newblock {\em J. ACM}, 27(2):356–364, 1980.
\newblock \href {https://doi.org/10.1145/322186.322198} {\path{doi:10.1145/322186.322198}}.

\bibitem[N{\'{O}}89]{Narendran1989}
Paliath Narendran and Colm {\'{O}}'D{\'{u}}nlaing.
\newblock {Cancellativity in Finitely Presented Semigroups}.
\newblock {\em J. Symb. Comput.}, 7(5):457--472, 1989.
\newblock \href {https://doi.org/10.1016/S0747-7171(89)80028-8} {\path{doi:10.1016/S0747-7171(89)80028-8}}.

\bibitem[Ohl94]{Ohlebusch1994}
Enno Ohlebusch.
\newblock On the {M}odularity of {T}ermination of {T}erm {R}ewriting {S}ystems.
\newblock {\em Theor. Comput. Sci.}, 136(2):333--360, 1994.
\newblock \href {https://doi.org/10.1016/0304-3975(94)00039-L} {\path{doi:10.1016/0304-3975(94)00039-L}}.

\bibitem[PS81]{Peterson1981}
Gerald~E. Peterson and Mark~E. Stickel.
\newblock {Complete Sets of Reductions for Some Equational Theories}.
\newblock {\em J. {ACM}}, 28(2):233--264, 1981.
\newblock \href {https://doi.org/10.1145/322248.322251} {\path{doi:10.1145/322248.322251}}.

\bibitem[Rub95]{Rubio1995}
Albert Rubio.
\newblock {Theorem Proving modulo Associativity}.
\newblock In Hans~Kleine B{\"{u}}ning, editor, {\em Computer Science Logic, 9th International Workshop, {CSL} '95, Annual Conference of the EACSL, Paderborn, Germany, September 22-29, 1995, Selected Papers}, volume 1092 of {\em Lecture Notes in Computer Science}, pages 452--467. Springer, 1995.
\newblock \href {https://doi.org/10.1007/3-540-61377-3_53} {\path{doi:10.1007/3-540-61377-3_53}}.

\bibitem[Rus87]{Rusinowitch1987}
Micha{\"{e}}l Rusinowitch.
\newblock On {T}ermination of the {D}irect {S}um of {T}erm-{R}ewriting {S}ystems.
\newblock {\em Inf. Process. Lett.}, 26(2):65--70, 1987.
\newblock \href {https://doi.org/10.1016/0020-0190(87)90039-1} {\path{doi:10.1016/0020-0190(87)90039-1}}.

\bibitem[Sim94]{Sims1994}
C.~C. Sims.
\newblock {\em Computation with finitely presented groups}.
\newblock Cambridge University Press, New York, NY, 1994.
\newblock \href {https://doi.org/10.1017/CBO9780511574702} {\path{doi:10.1017/CBO9780511574702}}.

\bibitem[SS82]{Siekmann1982}
Jorg Siekmann and P~Szab{\'o}.
\newblock A {N}oetherian and confluent rewrite system for idempotent semigroups.
\newblock {\em Semigroup Forum}, 25(1):83--110, 1982.
\newblock \href {https://doi.org/10.1007/BF02573590} {\path{doi:10.1007/BF02573590}}.

\bibitem[Str98]{Polina1998}
Polina Strogova.
\newblock Finding a {F}inite {G}roup {P}resentation {U}sing {R}ewriting.
\newblock In Manuel Bronstein, Volker Weispfenning, and Johannes Grabmeier, editors, {\em Symbolic Rewriting Techniques}, pages 267--276, Basel, 1998. Birkh{\"a}user Basel.
\newblock \href {https://doi.org/10.1007/978-3-0348-8800-4_13} {\path{doi:10.1007/978-3-0348-8800-4_13}}.

\bibitem[SW15]{Sjoberg2015}
Vilhelm Sj\"{o}berg and Stephanie Weirich.
\newblock Programming up to {C}ongruence.
\newblock {\em SIGPLAN Not.}, 50(1):369–382, 2015.
\newblock \href {https://doi.org/10.1145/2676726.2676974} {\path{doi:10.1145/2676726.2676974}}.

\bibitem[Toy87]{Toyama1987}
Yoshihito Toyama.
\newblock On the {C}hurch-{R}osser property for the direct sum of term rewriting systems.
\newblock {\em J. {ACM}}, 34(1):128--143, 1987.
\newblock \href {https://doi.org/10.1145/7531.7534} {\path{doi:10.1145/7531.7534}}.

\end{thebibliography}
\end{document}